\documentclass[11pt]{article}
\usepackage[utf8]{inputenc}
\usepackage{colortbl}
\usepackage{fullpage}
\usepackage[margin = 2.5cm]{geometry}
\usepackage{microtype}
\usepackage{epsfig}
\usepackage{microtype}
\usepackage{graphicx}
\usepackage{enumitem}
\usepackage{multirow}
\usepackage{epsfig,enumerate,amsmath,amsfonts,amssymb,amsthm,mathrsfs,ifpdf}
\usepackage{indentfirst,relsize}
\usepackage[numbers]{natbib}
\usepackage{setspace}
\usepackage[symbol]{footmisc}

\usepackage{enumerate}
\usepackage{latexsym}
\usepackage{stackrel}
\usepackage[vlined,ruled]{algorithm2e}
\usepackage[all]{xy}
\usepackage[usenames,dvipsnames]{pstricks}
\usepackage{pst-grad} 
\usepackage{pst-plot} 
\usepackage{xspace}

\usepackage{lineno}
\newcommand{\size}[1]{\left| #1 \right|}
\newcommand{\E}{\mathbb{E}}

\newcommand{\remove}[1]{}

\newcommand{\N}{\mathbb{N}}

\newcommand{\cS}{\mathcal{S}}

\newcommand{\cA}{\mathcal{A}}

\newcommand{\eps}{\epsilon}
\newcommand{\vareps}{\varepsilon}

\newcommand{\complain}[1]{\textcolor{red}{#1}}
\newcommand{\comments}[1]{\textcolor{blue}{\bf{#1}}}
\newcommand{\pr}{{\mathbb{P}}\xspace}

\theoremstyle{plain}
\newtheorem{theo}{Theorem}[section]
\newtheorem{lem}[theo]{Lemma}

\newtheorem{cl}[theo]{Claim}

\theoremstyle{definition}

\newtheorem{rem}{Remark}
\newtheorem{obs}[theo]{Observation}

\newtheorem{problem}[theo]{Problem}

\newcommand{\ealong}{{\sc Edge Arrival}\xspace}
\newcommand{\ea}{{\sc EA}\xspace}

\newcommand{\valong}{{\sc Vertex Arrival}\xspace}
\newcommand{\va}{{\sc VA}\xspace}
\newcommand{\vadeglong}{{\sc Vertex Arrival with Degree Oracle}\xspace}
\newcommand{\vadeg}{{\sc VAdeg}\xspace}
\newcommand{\varandlong}{{\sc Vertex Arrival in Random Order}\xspace}
\newcommand{\varand}{{\sc VArand}\xspace}
\newcommand{\allong}{{\sc Adjacency List}\xspace}
\newcommand{\al}{{\sc AL}\xspace}
\newcommand{\colorverify}{{\sc Conflict-Sep}\xspace}
\newcommand{\conflictest}{{\sc Conflict-Est}\xspace}

\newcommand{\hmR}{$\mbox{high-m}_R$ }
\newcommand{\lmR}{$\mbox{low-m}_R$ }

\usepackage{xcolor}
\usepackage{mdframed}
\usepackage{hyperref}

\title{
Even the Easiest(?) Graph Coloring Problem is not Easy in Streaming!
}
\date{}
{
\author{
Anup Bhattacharya \footnote{Funded by NPDF fellowship at ISI, Kolkata}
\and
Arijit Bishnu
\footnote{
Indian Statistical Institute, Kolkata, India
}
\and
Gopinath Mishra
\footnotemark[1]
\and
Anannya Upasana
\footnotemark[1]
}}
\begin{document}

\maketitle

\thispagestyle{empty}

\remove{
Works to be done: 
\begin{itemize}
 \item Fix a title
 \item Check all the colored text in the main body of the paper.
 \item Change the numbering system in footnote to symbols.
 \item Write a conclusion and discussion. -- Done! Please see once. 
 \item Mention in Appendix~\ref{sec:est-vaubsc} about the confusion of existence of independence in applying Chernoff. -- Done! See Step-(iv) and Remark~\ref{rem:why-ind}.  I will verify it once more tomorrow -- Arijit.
 \item Try to fit Algorithm~\ref{algo:random-order} in one page. -- Done!
 \item Bring Algorithms~\ref{algo:random-order} and \ref{algo:random-order-verify} in same font.
 \item To take care of any relevant ISAAC comments/suggestions. -- W-streaming defined. Anything more?
 \item Shall we mention whether the space complexity is deterministic/randomized? -- Probably No!
 \item Mention about that the random order lower bound work for AL model and with degree oracle. -- To be done!
\end{itemize}
}

\begin{abstract}
We study a graph coloring problem that is otherwise easy but becomes quite non-trivial in the one-pass streaming model. 
In contrast to previous graph coloring problems in streaming that try to find an assignment of colors to vertices, our main work is on estimating the number of conflicting or monochromatic edges given a coloring function that is streaming along with the graph; we call the problem {\sc Conflict-Est}. The coloring function on a vertex can be read or accessed only when the vertex is revealed in the stream. If we need the color on a vertex that has streamed past, then that color, along with its vertex, has to be stored explicitly. 
We provide algorithms for a graph that is streaming in different variants of the one-pass vertex arrival streaming model, viz. the {\sc Vertex Arrival} ({\sc VA}), {Vertex Arrival With Degree Oracle} ({\sc VAdeg}), {\sc Vertex Arrival in Random Order} ({\sc VArand}) models, with special focus on the random order model. We also provide matching lower bounds for most of the cases. The mainstay of our work is in showing that the properties of a random order stream can be exploited to design streaming algorithms for estimating the number of conflicting edges. We have also obtained a lower bound, though not matching the upper bound, for the random order model. Among all the three models vis-a-vis this problem, we can show a clear separation of power in favor of the \varand model.
\remove{
We give algorithms for \conflictest and \colorverify. For the \valong (\va) model, we need $\widetilde{\mathcal{O}}\left(\min\{\size{V},\frac{|V|^2}{T}\}\right)$ space to solve the \conflictest and $\widetilde{\mathcal{O}}\left(\frac{|V|}{\sqrt{|E|}}\right)$ space to solve the \colorverify problem. When given access to a degree oracle (\vadeg model), the space required to estimate the number of conflicting edges is $\widetilde{\mathcal{O}}\left(\frac{|E|}{T}\right)$ and to verify the coloring function is $\widetilde{\mathcal{O}}\left(\frac{1}{\vareps}\right)$. When the vertices arrive in a random order, we are able to give a $\widetilde{\mathcal{O}}\left(\frac{|V|}{\sqrt{|T|}}\right)$ space algorithm to solve the estimation and verification problems. Our results also follow in the \allong (\al) model. 
We also establish matching lower bounds for the easier variant \colorverify proving that our algorithms are efficient.
}
\end{abstract}
{\bf{Keywords:} Streaming, Graph coloring, Sublinear Algorithms.}

\newpage
\pagestyle{plain}
\setcounter{page}{1}

\section{Introduction}\label{sec:intro}
\noindent
\remove{Graph coloring is a fundamental problem both in graph theory and algorithms.} The \emph{chromatic number} $\chi(G)$ of an $n$-vertex graph $G=(V,E)$ is the minimum number of colors needed to color the vertices of $V$ so that no two adjacent vertices get the same color. The \emph{chromatic number} problem is NP-hard and even hard to approximate within a factor of $n^{1-\vareps}$ for any constant $\vareps > 0$~\cite{DBLP:journals/jcss/FeigeK98, DBLP:journals/toc/Zuckerman07, DBLP:conf/icalp/KhotP06}. For any connected undirected graph $G$ with maximum degree $\Delta$, $\chi(G)$ is at most $\Delta + 1$~\cite{Vizing1964OnAE}. This existential coloring scheme can be made constructive across different models of computation. 
A seminal result of recent vintage is that the $\Delta+1$ coloring can be done in the streaming model~\cite{DBLP:conf/soda/AssadiCK19}. Of late, there has been interest in graph coloring problems in the sub-linear regime across a variety of models~\cite{approx/AlonA20, DBLP:conf/soda/AssadiCK19, DBLP:conf/esa/BehnezhadDHKS19, DBLP:journals/corr/abs-1807-07640, DBLP:journals/corr/abs-1905-00566}. Keeping with the trend of coloring problems, these works look at assigning colors to vertices. Since the size of the output will be as large as the number of vertices, reseachers study the semi-streaming model~\cite{DBLP:journals/sigmod/McGregor14} for streaming graphs. In the semi-streaming model, $\widetilde{\mathcal{O}}(n)$\footnote{$\widetilde{\mathcal{O}}(\cdot)$ hides a polylogarithmic factor.} space is allowed. 

In a marked departure from the above works that look at the classical coloring problem, the starting point of our work is (inarguably?) the simplest question one can ask in graph coloring -- given a coloring function $f: V \rightarrow \{1, \ldots, C\}$ on the vertex set $V$ of a graph $G=(V,E)$, is $f$ a valid coloring, i.e., for any edge $e \in E$, do both the endpoints of $e$ have different colors? This is the problem one encounters while proving that the problem of chromatic number belongs to the class {\sf NP}~\cite{DBLP:books/fm/GareyJ79}. 
\conflictest, the problem of estimating the number of monochromatic (or, conflicting) edges for a graph $G$ given a coloring function $f$, remains a simple problem in the {\sc RAM} model; it even remains simple in the one-pass streaming model if the coloring function $f$ is marked on a \emph{public board}, readable at all times. We show that the problem throws up interesting consequences if the coloring function $f$ on a vertex is revealed only when the vertex is revealed in the stream. For a streaming graph, if the vertices are assigned colors arbitrarily or randomly on-the-fly while it is exposed, our results can also be used to estimate the number of conflicting edges. 
 These problems also find their use in estimating the number of conflicts in a job schedule and verifying a given job schedule in a streaming setting. This can also be extended to problems in various domains like frequency assignment in wireless mobile networks and register allocation~\cite{DBLP:journals/scheduling/EvenHKR09}. 
As the problem, by its nature, admits an estimate or a yes-no answer, the need of the space to store all vertices as in the semi-streaming model goes away and we can focus on space efficient algorithms in the conventional graph streaming models like \valong~\cite{DBLP:conf/icalp/CormodeDK19}. We also note in passing that many of the trend setting problems in streaming, like frequency moments, distinct elements, majority, etc.\ have been simple problems in the ubiquitous {\sc RAM} model as the coloring problem we solve here. 

\section{Preliminaries}
\label{sec:prelim}
\subsection{Notations and the streaming models}
\noindent 
{\bf Notations.} We denote the set $\{1,\ldots,n\}$ by $[n]$. $G(V(G),E(G))$ denotes a graph where $V(G)$ and $E(G)$ denote the set of vertices and edges of $G$, respectively; $\size{V}=n$ and $\size{E}=m$. We will write only $V$ and $E$ for vertices and edges when the graph is clear from the context.  We denote $E_M \subseteq E$ as the set of monochromatic edges.
\remove{
For the \conflictest problem, we assume a lower bound on $E_M$ and call it $T$, where $1 \leq T \leq {n \choose 2}$. The input graph $G$  is promised to satisfy $\size{E_M}\geq T$. 
}
The set of neighbors of a vertex $u \in V(G)$ is denoted by $N_G(u)$ and the degree of a vertex $u \in V(G)$ is denoted by $d_G(u)$.
Let $N_G(u) = N^{-}_G(u) \uplus N^{+}_G(u)$ where $N^{-}_G(u)$ and $N^{+}_G(u)$ denote the set of neighbors of $u$ that have been exposed already and are yet to be exposed, respectively in the stream.
Also, $d_G(u) = d^{-}_G(u) + d^{+}_G(u)$ where $d^{-}_G(u) = \size{N^{-}_G(u)}$ and $d^{+}_G(u) = \size{N^{+}_G(u)}$. For a monochromatic edge $(u,v) \in E_M$, we refer to $u$ and $v$ as monochromatic neighbors of each other. We define $d_M(u)$ to be the number of monochromatic neighbors of $u$ and hence, the monochromatic degree of $u$. 
\remove{ \complain{So, $\size{E_M} =\frac{1}{2}\sum_{u \in V} d_M(u)$. Arijit: Why is it needed in the notations section?}}

\remove{
Let $\mathbb{E}[X]$ and $\mathbb{V}[X]$ \complain{(Do we need variance anywhere?)} denote the expectation and variance of the random variable $X$.} Let $\mathbb{E}[X]$ denote the expectation of the random variable $X$. For an event $\mathcal{E}$, $\overline{\mathcal{E}}$ denotes the complement of $\mathcal{E}$. 
$\mathbb{P}(\mathcal{E})$ denotes the probability of an event $\mathcal{E}$.
The statement ``event $\mathcal{E}$ occurs with high probability'' is equivalent to $\mathbb{P}(\mathcal{E}) \geq 1-\frac{1}{n^c}$, where $c$ is an absolute constant. The statement ``$a$ is a $1 \pm \varepsilon$ multiplicative approximation of $b$'' means $|b-a| \leq \varepsilon \cdot b$. For $x \in \mathbb{R}$, $\exp(x)$ denotes the standard exponential function, that is, $e^x$. By polylogarithmic, we mean $\mathcal{O}\left( \left({\log n}/{\varepsilon}\right)^{\mathcal{O}(1)}\right)$. The notation $\widetilde{\mathcal{O}}(\cdot)$ hides a polylogarithmic term in $\mathcal{O}(\cdot)$.

\medskip

\noindent
{\bf Streaming models for graphs.} As alluded to earlier, the crux of the problem depends on the way the coloring function $f$ is revealed in the stream. The details follow.

\noindent
(i) {\valong (\va):} The vertices of $V$ are exposed in an arbitrary order. After a vertex $v \in V$ is exposed, all the edges between $v$ and pre-exposed neighbors of $v$, are revealed. This set of edges are revealed one by one in an arbitrary order. Along with the vertex $v$, only the color $f(v)$ is exposed, and not the colors of any pre-exposed vertices. So, we can check the monochromaticity of an edge $(v,u)$ only if $u$ and $f(u)$ are explicitly stored. \\
(ii) {\vadeglong (\vadeg)~\cite{DBLP:conf/pods/McGregorVV16, DBLP:conf/pods/BeraS20}:} This model works same as the \va model in terms of exposure of the vertex $v$ and the coloring on it; but we are allowed to know the degree $d_G(v)$ of the currently exposed vertex $v$ from a degree oracle on $G$. \\
(iii) {\varandlong (\varand)~\cite{DBLP:conf/kdd/StantonK12, DBLP:conf/wsdm/TsourakakisGRV14}:} 
\remove{
This model works same as the \va model but the vertices are revealed in a random order, i.e., the revealing order of the vertices is equally likely to be any one of the permutations of the vertices.}
This model works same as the \va model but the vertex sequence revealed is equally likely to be any one of the permutations of the vertices. \\
(iv) {\ealong (\ea):} The stream consists of edges of $G$ in an arbitrary order. As the edge $e$ is revealed, so are the colors on its endpoints. Thus the conflicts can be easily checked. \\
(v) {\allong (\al):} The vertices of $V$ are exposed in an arbitrary order. When a vertex $v$ is exposed, all the edges that are incident to $v$, are revealed one by one in an arbitrary order. Note that in this model each edge is exposed twice, once for each exposure of an incident  vertex. As in the \va model, here also only $v$'s color $f(v)$ is exposed.

As the conflicts can be checked easily in the \ea model in $O(1)$ space, a logarithmic counter is enough to count the number of monochromatic edges. The \al model works almost the same as the \vadeg model. So, we focus on the three models -- \va, \vadeg and \varand in this work and show that they have a clear separation in their power vis-a-vis the problem we solve. A crucial takeaway from our work is that the random order assumption on exposure of vertices has huge improvements in space complexity.

\subsection{Problem definitions, results and the ideas}\label{sec:ideas}
\noindent 
{\bf Problem definition.} Let the vertices of $G$ be colored with a function $f:V(G) \rightarrow [C]$, for $C \in \N$. An edge $(u,v) \in E(G)$ is said to be \emph{monochromatic} or \emph{conflicting} with respect to $f$ if $f(u)=f(v)$. A coloring function $f$ is called \emph{valid} if no edge in $E(G)$ is monochromatic with respect to $f$. For a given parameter $\vareps \in (0,1)$, $f$ is said to be $\vareps$-far from being \emph{valid} if at least $\vareps \cdot \size{E(G)}$ edges are monochromatic with respect to $f$. We study the following problems. \remove{that involve estimating the number of monochromatic edges in a graph arriving over a stream. Formally, the problems are stated as follows.}
\begin{problem}[{\sc Conflict Estimation} aka \conflictest] A graph $G=(V,E)$ and a coloring function $f:V(G) \rightarrow [C]$ are streaming inputs. Given an input parameter $\vareps > 0$, the objective is to estimate the number of monochromatic edges in $G$ within a $(1\pm \vareps)$-factor.
\label{prob:estimate}
\end{problem}
\begin{problem}[{\sc Conflict Separation} aka \colorverify] A graph $G=(V,E)$ and a coloring function $f:V(G) \rightarrow [C]$ are streaming inputs. Given an input parameter $\vareps > 0$, the objective is to distinguish if the coloring function $f$ is valid or is $\vareps$-far from being valid. 
\label{prob:separate}
\end{problem}
\begin{rem} 
Problem~\ref{prob:estimate} is our main focus, but we will mention a result on Problem~\ref{prob:separate} in Section~\ref{sec:sep-vaub}. Notice that \conflictest is a \emph{difficult} problem than \colorverify.
\end{rem}

\noindent
{\bf The results and the ideas involved.}
All our upper and lower bounds on space are for one-pass streaming algorithms. Table~\ref{table:results-conflict-and-est-graph} states our results for the \conflictest problem, the main problem we solve in this paper, across different variants of the \va model. The main thrust of our work is on estimating monochromatic edges under random order stream. For random order stream, we present both upper and lower bounds in Sections~\ref{sec:est-vardubsc} and \ref{sec:lower-rand}. There is a gap between the upper and lower bounds in the \varand model, though we have a strong hunch that our upper bound is tight. Apart from the above, using a structural result on graphs, we show in Section~\ref{sec:sep-vaub} that the \colorverify problem admits an easy algorithm in the \varand model. To give a complete picture across different variants of the \va models,  we show matching upper and lower bounds for the \va and \vadeg models in Section~\ref{sec:est-vaubsc} and Appendix~\ref{sec:est-lowerbound}.

\small
\begin{table}[thb]
\centering
\begin{tabular}{c || c | c | c } 

 
  Model & \va & \varand & \vadeg \\
 \cline{1-4}
       \hline \hline
  \multirow{2}{*} {Upper Bound} & $\widetilde{\mathcal{O}}\left(\min\{\size{V},\frac{|V|^2}{T}\}\right)$ & $\widetilde{\mathcal{O}}\left(\frac{|V|}{\sqrt{T}}\right)$ & $\widetilde{\mathcal{O}}\left(\min \{\size{V},\frac{|E|}{T}\}\right)$ \\
  & (Sec.~\ref{sec:est-vaubsc}, Thm.~\ref{theo:est-vaubth}) & (Sec.~\ref{sec:est-vardubsc}, Thm.~\ref{theo:est-vardubth}) & (Sec.~\ref{sec:est-vaubsc}, Thm.~\ref{theo:est-vadgubth}) \\
  \cline{1-4}
  \multirow{2}{*} {Lower Bound} & $\Omega\left(\min\{\size{V},\frac{|V|^2}{T}\}\right)$ & $\Omega \left(\frac{\size{V}}{T^2} \right)$ & $\Omega\left(\min\{\size{V},\frac{|E|}{T}\}\right)$ \\
  & (Sec.~\ref{sec:est-vaaolbsc}, Thm.~\ref{theo:est-vaaolbth}) & (Sec.~\ref{sec:lower-rand}, Thm.~\ref{thm:rand-lower}) & (Sec.~\ref{sec:est-vadglbsc}, Thm.~\ref{theo:est-vadglbth}) \\
  \hline
  
\end{tabular}
\caption{This table shows our results on \conflictest on a graph $G(V,E)$ across different \valong models. Here, $T > 0$ denotes the promised lower bound on the number of monochromatic edges.}
\label{table:results-conflict-and-est-graph}
\end{table} 

\normalsize

The promise $T$ on the number of monochromatic edges is a very standard assumption for estimating substructures in the world of graph streaming algorithm~\cite{DBLP:conf/focs/KallaugherKP18,DBLP:conf/icalp/KaneMSS12,DBLP:conf/pods/KallaugherMPV19, DBLP:conf/pods/McGregorVV16, DBLP:conf/stacs/BeraC17}.~\footnote{Here we have cited a few. However, there are huge amount of relevant literature.}

We now briefly mention the salient ideas involved. For the simpler variant of \conflictest in \va model, we first check if $\size{V} \geq T$. If yes, we store all the vertices and their colors in the stream to determine the exact value of the number of monochromatic edges. Otherwise, we sample each pair of vertices $\{u,v\}$ in ${V \choose 2}$~\footnote{${V \choose 2}$ denotes the set of all size 2 subsets of $V(G)$.}, with probability $\widetilde{\mathcal{O}}\left({1}/{T}\right)$ independently~\footnote{Note that we might sample some pairs that are not forming edges in the graph.} before the stream starts. {When the stream comes, we compute the number of monochromatic edges from this sample.} The details are in Section~\ref{sec:est-vaubsc}.
\remove{
For the simpler variant of \conflictest in \va model in Section~\ref{sec:est-vaubsc}, we sample a \emph{suitable} number of pairs of vertices before the stream starts, compute the number of monochromatic edges in the sample and use it to estimate the number of monochromatic edges in the graph.} Though the algorithm looks extremely simple, it matches the lower bound result for \conflictest in \va model, presented in  Appendix~\ref{sec:est-lowerbound}. The \vadeg model with its added power of a degree oracle, allows us to know $d_G(u)$ for a vertex $u$ and as edges to pre-exposed vertices are revealed, we also know $d^{-}_G(u)$ and $d^{+}_G(u)$. This allows us to use sampling to store vertices and to use a technique which we call \emph{sampling into the future} 
where indices of random neighbors, out of $d^{+}_G(u)$ neighbors, are selected for future checking. The upper bound result, 
for \conflictest in \vadeg model, is presented in Section~\ref{sec:est-vaubsc}, and it is tight as we prove a matching lower bound in Appendix~\ref{sec:est-lowerbound}.

The algorithm for \conflictest in \varand model is the mainstay of our work and is presented in Section~\ref{sec:est-vardubsc}. We redefine the degree in terms of the number of monochromatic neighbors a vertex has in the randomly sampled set. Here, we estimate the high monochromatic degree and low monochromatic degree vertices separately by sampling a random subset of vertices.  While the monochromatic degree for the high degree vertices can be extrapolated from the sample, handling low monochromatic degree vertices individually in the same way does not work. To get around, we group such vertices having similar monochromatic degress and treat them as an entity. We also provide a lower bound for the \varand model, in Section~\ref{sec:lower-rand}, using a reduction from \emph{multi-party set disjointness}; though there is a gap in terms of the exponent in $T$.

\remove{
 The result in Section~\ref{sec:est-vardubsc} that solves \conflictest in the \varand model is the mainstay of our work and it uses the random order on the input. Here, we estimate the \emph{monochromatic degree} of high degree and low degree vertices separately by sampling a random subset of vertices. We redefine this degree in terms of the number of monochromatic neighbors a vertex has in the randomly sampled set. While the monochromatic degree for the high degree vertices can be extrapolated from the sample, we use a bucketing technique to estimate the same for the low degree vertices. We have been able to provide a lower bound for the \varand model using a reduction from multi-party set disjointness; though there is a gap in terms of the exponent in $T$.
}
 \remove{
 The structural property of the graph is exploited to give our upper bounds for the \colorverify problem in Section~\ref{sec:sep-vaub}. We consider the subgraph defined on the conflicting edges and claim the existence of either a large matching or a high degree vertex. In the latter case, if the high degree vertex appears at the beginning of the stream and we fail to store it, we lose the chance to sample conflicting edges. To remedy this, we make use of the random order. This ensures that the high degree vertex appears after a constant fraction of its neighbors have appeared, in turn ensuring their storage and enabling the check of conflicting edges.
}
 
The highlights of our work are as follows: 
\begin{itemize}
 \item We show that possibly the easiest graph coloring problem is worth studying over streams. 
 \item For researchers working in streaming, the \emph{gold standard} is the \ea model as most problems are non-trivial in this model. We point out a problem that is harder to solve in the \va model as compared to the \ea model. 
 \item We show that the three \va related models have a clear separation in their space complexities vis-a-vis the problem we solve. We could exploit the random order of the arrival of the vertices to get substantial improvements in space complexity. 
 \item We could obtain lower bounds for all the three models but the lower bounds are matching for the \va and \vadeg models.
\end{itemize}

\remove{
The structural property of the graph is exploited to give our upper bounds for the \colorverify problem in Section~\ref{sec:sep-vardubsc}. We consider the subgraph defined on the conflicting edges and claim the existence of either a large matching or a high degree vertex. In the latter case, if the high degree vertex appears at the beginning of the stream and we fail to store it, we lose the chance to sample conflicting edges. To remedy this, we make use of the random order. This ensures that the high degree vertex appears after a constant fraction of its neighbors have appeared, in turn ensuring their storage and enabling the check of conflicting edges.
}

\subsection{Prior works on graph coloring in semi-streaming model.} 
\label{ssec:prior-work}
\noindent
Bera and Ghosh~\cite{DBLP:journals/corr/abs-1807-07640} commenced the study of vertex coloring in the semi-streaming model. They devise a randomized one pass streaming algorithm that finds a $(1+ \vareps) \Delta$ vertex coloring in $\widetilde{\mathcal{O}}(n)$ space. 
\remove{
They do this in two phases by first randomly partitioning the vertex set into $\mathcal{O}(\frac{\Delta}{\log{n}})$ subsets where the subgraph induced by each subset has a maximum degree of $\log{n}$ with high probability. Then, every vertex of the random partitioning is colored independently and uniformly at random using a $\mathcal{O}(\frac{\Delta}{\log{n}})$ sized color palette.
}
Assadi et al.~\cite{DBLP:conf/soda/AssadiCK19} find a proper vertex coloring using $\Delta + 1$ colors via various classes of sublinear algorithms. Their state of the art contributions can be attributed to a key result called the \emph{palette-sparsification theorem} which states that for an $n$-vertex graph with maximum degree $\Delta$, if $\mathcal{O}(\log{n})$ colors are sampled independently and uniformly at random for each vertex from a list of $\Delta + 1$ colors, then with a high probability a proper $\Delta + 1$ coloring exists for the graph. They design a randomized one-pass dynamic streaming algorithm for the $\Delta + 1$ coloring using $\widetilde{\mathcal{O}}(n)$ space. The algorithm takes post-processing $\widetilde{\mathcal{O}}(n \sqrt{\Delta})$ time and assumes a prior knowledge of $\Delta$. Alon and Assadi~\cite{abs-2006-10456} improve the palette sparsification result of~\cite{DBLP:conf/soda/AssadiCK19}. They consider situations where the number of colors available is both more than and less than $\Delta + 1$ colors. They show that sampling $\mathcal{O}_{\vareps}(\sqrt{\log n})$ colors per vertex is sufficient and necessary for a $(1+\vareps)\Delta$ coloring. 
Bera et al.~\cite{DBLP:journals/corr/abs-1905-00566} give a new graph coloring algorithm in the semi-streaming model where the number of colors used is parameterized by the degeneracy $\kappa$. The key idea is a \emph{low degeneracy partition}, also employed in \cite{DBLP:journals/corr/abs-1807-07640}. The numbers of colors used to properly color the graph is $\kappa + o(\kappa)$ and post-processing time of the algorithm is improved to $\widetilde{\mathcal{O}}(n)$, without any prior knowledge about $\kappa$. 
Behnezhad et al.~\cite{DBLP:conf/esa/BehnezhadDHKS19} were the first to give one-pass W-streaming algorithms (streaming algorithms where outputs are produced in a streaming fashion as opposed to outputs given finally at the end) for edge coloring both when the edges arrive in a random order or in an adversarial fashion.

\remove{
Behnezhad et al.~\cite{DBLP:conf/esa/BehnezhadDHKS19} give a one-pass $\widetilde{\mathcal{O}}(n)$ space W-streaming algorithm that always returns a valid edge coloring and uses $\mathcal{O}(\Delta)$ colors when the edges arrive in a random order. The algorithm uses $5.44 \Delta$ colors with high probability. A W-streaming model is one in which the output is also reported in a streaming fashion. They give another one pass $\tilde{\mathcal{O}}(n)$ space W-streaming algorithm to solve the edge coloring problem using $\mathcal{O}(\Delta^{2})$ colors when the edges arrive in an adversarial fashion. 
}

\remove{\complain{Arijit: Should we not review \cite{abs-2006-10456}?}}
\section{\conflictest in \va and \vadeg models}\label{sec:est-vaubsc} 
\noindent
In this Section, we design algorithms for \conflictest problem in the \va and \vadeg models. We show matching lower bounds later in Appendix~\ref{sec:est-lowerbound}. Mainly, we prove the following two theorems here. 

 \begin{theo}\label{theo:est-vaubth}
 Given any graph $G=(V,E)$ and a coloring function $f: V \rightarrow [C]$ as input in the stream, there exists an algorithm that solves the \conflictest problem in the \va model with high probability in $\widetilde{\mathcal{O}} \left( \min \left( \size{V},\frac{|V|^2}{T} \right) \right) $ space, where $T$ is a lower bound on the number of monochromatic edges in the graph. 
 \end{theo}
\begin{theo}\label{theo:est-vadgubth}
Given any graph $G=(V,E)$ and a coloring function $f: V \rightarrow [C]$ as input in the stream, there exists an algorithm that solves the \conflictest problem in the \vadeg model with high probability in $\widetilde{\mathcal{O}}\left(\min\{\size{V}, \frac{\size{E}}{T}\}\right)$ space, where $T$ is a lower bound on the number of monochromatic edges in the graph.  
\end{theo}

Before going to the algorithms for \conflictest problem in the \va and \vadeg model, we discuss as a warm-up, a two-pass algorithm for \conflictest in the \va model that uses $\widetilde{\mathcal{O}}\left(\min\{\size{V}, \frac{\size{E}}{T}\}\right)$ space, where $T$ is the promised lower bound on the number of monochromatic edges in the graph. Here we assume that $|E|$ is known to the algorithm. However, this assumption can be removed easily in a setting with two passes.
\paragraph*{A two-pass algorithm for \conflictest in \va model (described informally):} 
\begin{description}
\item[If $T \leq \frac{\size{E}}{\size{V}}$:] Our algorithm stores all the vertices and their colors. Thus we can determine the number of monochromatic edges exactly. The algorithm in this case is one pass and uses $\widetilde{\mathcal{O}}(\size{V})$ space.
\item[If $T > \frac{\size{E}}{\size{V}}$:] In the first pass, store each edge with probability $\widetilde{\mathcal{O}}\left(\frac{1}{T}\right)$. In the second pass, we check each edge stored in the first pass for conflict. In this way, we determine the number of monochromatic edges in the sample, from which, we can obtain a desired approximation of the number of monochromatic edges in the graph. The space complexity of our algorithm in this case is $\widetilde{\mathcal{O}}\left( \frac{\size{E}}{T}\right)$.
\end{description} 
If only one pass is allowed, the above algorithm, when $T>\frac{\size{E}}{\size{V}}$, can not be simulated in \va model because of the following reason. Consider an edge $(u,v) \in E_M$ such that $u$ is exposed before $v$. Note that we will be able to know about the edge only when $v$ is exposed but we will be able to check whether $(u,v) \in E_M$ only when we have stored $u$ and its color. However, there is no clue about the edge $(u,v)$ when $u$ is exposed. So, to solve it in one-pass, we sample each pair of vertices {(without bothering if there is an edge between them)} with probability $\widetilde{\mathcal{O}}\left( \frac{1}{T}\right)$, before the start of the stream, and determine the number of monochromatic edges in the sample to get an estimate of the number of monochromatic edges in $G$. This implies that the space complexity of the algorithm for \conflictest in \va model is $\widetilde{\mathcal{O}}\left( \frac{\size{V}^2}{T}\right)$ as stated in Theorem~\ref{theo:est-vaubth}. In \vadeg model, when $u$ is exposed we will get $d_G(u)$ and hence $d^+_G(u)$. The degree information, when $u$ is exposed, gives some statistics regarding how the vertex $u$ might be useful in the future. We exploit this advantage of \vadeg model over \va model to get an algorithm for \conflictest that has better space complexity (See Theorem~\ref{theo:est-vadgubth}).

\subsection{Proof of Theorem~\ref{theo:est-vaubth}}
\label{sec:ub-vadeg}
\noindent Our algorithm for \conflictest for \va model, first checks if $T \leq \size{V}$. If yes, we store all the vertices along with their colors to estimate the number of monochromatic edges in the graph exactly. So, the space used by the algorithm is $\widetilde{\mathcal{O}}(\size{V})$ when $T \leq \size{V}$. We will be done by giving an algorithm for \conflictest in \va model that uses $\widetilde{\mathcal{O}}\left( \frac{\size{V}^2}{T}\right)$ space. This algorithm will only be executed when $T > \size{V}$.
 
Let $V=\{v_1,\ldots,v_n\}$ be the vertices of the graph. Our algorithm starts by generating a sample $Z$ of vertex pairs where each $\{v_i,v_j\}$ is added to $Z$, independently, with probability $\frac{30 \log n}{\vareps^2 T}$. Note that $Z$ is obtained before the start of the stream. Over the stream, we check the following for each $\{v_i,v_j\}\in Z$: whether $(v_i,v_j) \in E$ and is monochromatic.
\remove{
\begin{itemize}
\item whether $(v_i,v_j)$ is an edge,
\item whether $f(v_i)=f(v_j)$, that is, whether $(v_i,v_j)$ is a monochromatic edge.
\end{itemize}
}
Let $S \subseteq Z$ be the set of monochromatic edges in $Z$. Note that the expected value of $\size{S}$ is  given by $\E[\size{S}]=\frac{30 \log n}{\vareps^2 T}\size{E_M}$.

We report $\widehat{m}=\frac{\vareps^2 T}{30 \log n}\size{S}$ as our estimate for $\size{E_M}$. Applying Chernoff bound (See Lemma~\ref{lem:cher_bound1} in Appendix~\ref{sec:prob}), we guarantee that 
$$\pr \left(\size{\widehat{m}-\size{E_M}} \geq \vareps \size{E_M}\right)\leq \pr \left(\size{\size{S}-\E[\size{S}]} \geq \vareps \E[\size{S}]\right) \leq \exp{\left(\frac{-\E[\size{S}]\vareps^2}{3}\right)}\leq \frac{1}{n^{10}}.$$
Note that the last inequality holds as $\E[\size{S}]=\frac{30 \log n}{\vareps^2 T}\size{E_M}$ and $\size{E_M}\geq T$.

Observe that the space used by our algorithm is ${\cal O}(\size{Z})$ when $T > \frac{\size{E}}{\size{V}}$. Note that $\E[\size{Z}]=\frac{30 \log n}{\vareps^2 T}{n \choose 2}$. Applying Chernoff bound (See Lemma~\ref{lem:cher_bound1} in Appendix~\ref{sec:prob}), we can show that $\size{Z}=\widetilde{\mathcal{O}}\left( \frac{n^2}{T}\right)$ with high probability.

Putting together the space complexities of our algorithms for the case $T \leq \size{V}$ and $T > \size{V}$, we have the desired bound on the space.

\subsection{Proof of Theorem~\ref{theo:est-vadgubth}}\label{sec:ub-arbit}
\noindent
For simplicity of presentation, assume that we know the number of edges $|E|$ in the graph. We will discuss ways to remove this assumption later.
\subsubsection{Algorithm for \conflictest in \vadeg model when $|E|$ is known}\label{sec:e-known-vadeg}
\noindent
Our algorithm for \conflictest for \vadeg model, first checks if $T \leq \frac{\size{E}}{\size{V}}$. If $T \leq \frac{\size{E}}{\size{V}}$, we store all the vertices along with their colors to estimate the number of monochromatic edges in the graph exactly. So, the space used by the algorithm is $\widetilde{\mathcal{O}}(\size{V})$ when $T \leq \frac{\size{E}}{\size{V}}$. We will be done by giving an algorithm for \conflictest in \vadeg model that uses $\widetilde{\mathcal{O}}\left(\frac{\size{E}}{T}\right)$ space. This algorithm will be executed only when $T > \frac{\size{E}}{\size{V}}$.
 
Let $V=\{v_1,\ldots,v_n\}$ and w.o.l.g. the vertices are exposed in the order $v_1,\ldots,v_n$. However, our algorithm does not know about the ordering of the vertices  {in the stream}. Our algorithm stores the following information.
\begin{itemize}
\item  A random subset $Y \subset V \times [n]$ that will be generated over the stream;
\item  a subset $A$ of vertices formed from the first elements in the pairs present in $Y$; the colors of the vertices are also stored;
\item  for each vertex $v \in A$, a number $\ell_v$ that denotes the number of neighbors in $N_G^+(v)$ that have been exposed. So, $\ell_v$ is initialized to $0$ when $v$ gets exposed in the stream and is at most $\size{N_G^+(v)}$ at any instance of the stream;
\item a subset $S\subseteq E_M$ of the set of monochromatic edges in $G$.
\end{itemize}

When a vertex $v_j$ is exposed, our algorithm performs the following steps:
\begin{itemize}
\remove{\item[(i)] Set $d^-_G(v_j)$ equals to the number of neighbors that $v_j$ has in $\{v_1,\ldots,v_{j-1}\}$. Get $d_G(v_j)$ from 
the degree oracle and compute $d_G^+(v_j)$;}
\item[(i)] Get $d_G(v_j)$ from the degree oracle and $d_G^-(v_j)$ from the exposed edges and compute $d_G^+(v_j)$;
\item[(ii)] Add $(v_j,k), k\in \left[d_G^+(v_j)\right]$, with probability $\frac{30 \log n}{\vareps^2 T}$ to $Y$, independently;
\item[(iii)] Add $v_j$ along with its color to $A$ if at least one  $(v_j,k)$ is added to $Y$. 
\item[(iv)] For each $v_i \in A$ such that $(v_i,v_j) \in E$, increment $\ell_{v_i}$ by $1$. 
\item[(v)] For each $v_i \in A$ such that $(v_i,\ell_{v_i}) \in Y$, check whether $(v_i,v_j)$ forms a monochromatic edge. If yes, add $(v_i,v_j)$ to $S$. {}{(This step ensures independence so that Chernoff bounds can be used. See Remark~\ref{rem:why-ind} below.)}  
\end{itemize}
The main catch of the algorithm for \conflictest in \vadeg model is in Step-(ii). Due to the added power of degree oracle, we are able to sample edges that have not arrived explicitly in the stream. We referred to this phenomenon as \emph{sampling into the future} in Section~\ref{sec:ideas}.
 
At the end of the stream, we report $\widehat{m}=\frac{\vareps^2 T}{30 \log n}\size{S}$ as the estimate of $\size{E_M}$. Now, we show that 
$\pr\left(\size{\widehat{m}-\size{E_M}}\geq \vareps \size{E_M}\right) \leq \frac{1}{n^{10}}.$
Consider a monochromatic edge $(v_i,v_j)\in E_M$. W.l.o.g., assume that $v_j$ is exposed sometime after $v_i$ is exposed in the stream. Let $r \in \left[d_G^+(v_i)\right]$ be such that $v_i$ has $r-1$ neighbors in $\{v_{i+1},\ldots,v_{j-1}\}$. So, $v_j$ is the $r$-th neighbor of $v_i$ exposed after the exposure of $v_i$. From the description of the algorithm, $(v_i,v_j)$ is added to $S$ if and only if $(v_i,r)$ is added to $Y$. Note that $(v_i,r)$ can be added to $Y$ only when the vertex $v_i$ is exposed in the stream. Before calculating $E[\size{S}]$ and applying Chernoff bound, we focus on the following remark.

\begin{rem}
{}{At the first look, it might appear that the monochromatic edges are not independently added to $S$. For example, let us consider the following situation. Let $(v_i,r')$, with $r'\in 
\left[d_G^+(v_i)\right]$ and $r' \neq r$, is added to $Y$, that is, $v_i$ is present in $A
$ and the color of $v_i$ is stored. So, when $v_j$ gets exposed along 
with its color, we can check whether $(v_i,v_j)$ is monochromatic 
irrespective of $(v_i,r)$ being added to $Y$. But the crucial point is 
that we add $(v_i,v_j)$ to $S$ only when $(v_i,r)$ is added to $Y$. 
However, $(v_k,\ell)$s, with $k\in [n]$ and $\ell \in \left[d_G^+(v_k)
\right]$, are added to $Y$, independently. That is, each monochromatic edge 
in $E_M$ is added to $S$, independently.}
\label{rem:why-ind}
\end{rem}
The probability that a monochromatic edge is added to $S$ is $\frac{30 \log n}{\vareps^2 T}$. That is, $\E[\size{S}]=\frac{30 \log n}{\vareps^2 T} \size{E_M}$.
  Applying Chernoff bound (See 
Lemma~\ref{lem:cher_bound1} in Appendix~\ref{sec:prob}), we can guarantee that 
$$\pr \left(\size{\widehat{m}-\size{E_M}} \geq \vareps \size{E_M}\right)\leq \pr \left(\size{\size{S}-\E[\size{S}]} \geq \vareps \E[\size{S}]\right) \leq \exp{\left(\frac{-\E[\size{S}]\vareps^2}{3}\right)}\leq \frac{1}{n^{10}}.$$
Note that the last inequality holds as $\size{E_M} \geq T$.
Observe that the space used by the algorithm is $\widetilde{\mathcal{O}}(\size{Y}+\size{A}+\size{S})=\widetilde{\mathcal{O}}({\size{Y}})$. Note that $\E[\size{Y}]=\sum\limits_{i=1}^n d_G^+(v_i)\cdot \frac{30 \log n}{\vareps^2 T}=\frac{30 \size{E}\log n}{\vareps^2 T}$. Applying Chernoff bound (See Lemma~\ref{lem:cher_bound1} in Appendix~\ref{sec:prob}), we can say that $\size{Y}=\widetilde{\mathcal{O}}\left( \frac{\size{E}}{T}\right)$ with high probability.
Putting together the space complexities of our algorithms for the case $T \leq \frac{\size{E}}{\size{V}}$ and $T > \frac{\size{E}}{\size{V}}$, we have the desired bound on the space.

\subsubsection{Modifying the algorithm in Section~\ref{sec:e-known-vadeg} when $\size{E}$ is unknown}
\noindent
In the modified algorithm, we maintain a counter defined as follows.
 $$\mbox{\sc cnt}:=\sum\limits_{v~\mbox{has been exposed}} d_G^+(v).$$
 Consider the following observation about {\sc cnt} that will be used in our analysis. As mentioned earlier, $\size{V}=n$.

 \begin{obs}\label{obs:cnt}
 At any point of the streaming algorithm, {\sc cnt} is a lower bound on $\size{E}$, the number of edges in the graph. Moreover, at the end of the stream, {\sc cnt} becomes $\size{E}$. {Also, {\sc cnt} is non-decreasing.}
 \end{obs} 
  We process the stream by maintaining $Y, \, A$ and $S$, as defined in the algorithm in Section~\ref{sec:e-known-vadeg}, for the case $T> \frac{|E|}{|V|}$, until {\sc cnt} reaches $\tau=100 \size{V} T \log n$, with a slight difference. Here, we add each $(v_j,\ell)$ to $Y$ with probability $\frac{3000\log n}{\vareps^3T}$ instead of $\frac{30\log n}{\vareps^2 T}$ as in Section~\ref{sec:e-known-vadeg}, where $v_j$ is a vertex exposed {while {\sc cnt} is less than $\tau$} and $\ell \in \left[d_G^+(v_j)\right]$. So, we have the following observation that will be used later in our analysis.
\begin{obs}\label{obs:degY}
{
With high probability, $\size{Y}=\widetilde{\mathcal{O}}\left(\size{V}\right)$ for all the instances in the stream while {\sc cnt} is less than $\tau$.}
\end{obs}  
\begin{proof}
{Let $v_k$ be the first exposed vertex in the stream when {\sc cnt} is more than $\tau$.} Also, let $U=\sqcup_{j=1}^{k-1}\{(v_j,\ell):\ell \in [d_G^+(v_j)]\}$, {where $\sqcup$ denotes disjoint union}. Observe that $\size{U}=\sum\limits_{j=1}^{k-1} d_G^+(v_j)<\tau$. We construct $Y$ by selecting independently each element of $U$ with probability $\frac{3000\log n}{\vareps^3T}$. Recall that $\tau=100|V|T\log n$. So, $\E\left[|Y|\right]=\frac{3000|U|\log n}{\vareps^3T}<\frac{3000\tau  \log n}{\vareps^3T}=\frac{300000(\log n)^2|V|}{\vareps^3}$. The observation follows by applying Chernoff bound (see Lemma~\ref{lem:cher_bound1} (iii) in Appendix~\ref{sec:prob}).
\end{proof}
   However, the modified algorithm behaves differently once {\sc cnt} is more than $\tau = 100 \size{V} T \log n$. \remove{Let {\sc cnt} reaches $\tau$ when vertex $v_k$ is exposed} {Let $v_k$ be as defined earlier}. We maintain two extra objects, as described below, after {\sc cnt} crosses $\tau$. 
 \begin{itemize}
 \item The set of vertices $B=\{v_k,\ldots,v_n\}$ and their colors;
 \item A counter $C_{>\tau}$ that denotes the number of monochromatic edges having both the endpoints in $B$.
 \end{itemize}
 The formal description of the modified algorithm is presented in Algorithm~\ref{algo:random-vadeg}.
\begin{algorithm}[t]\label{algo:random-vadeg}
\caption{{\sc Conflict-Est-Deg}($\vareps,T)$: \conflictest in \vadeg model}
\KwIn{$G = (V, E)$ and a coloring function $f$ on $V$ in the \vadeg model, parameters $T$ and $\vareps$ where $\vareps, T \geq 0$.}
\KwOut{$\widehat{m}$, that is, a $(1\pm \vareps)$ approximation to $\size{E_M}$.}
\For{(each exposed vertex $v_j$)}
{
\begin{itemize}
\item For each $v_i \in A$ such that $(v_i,v_j) \in E$, 
increment $\ell_{v_i}$ by $1$;
 \item For each $v_i \in A$ such that $(v_i,\ell_{v_i}) \in Y$, check whether $(v_i,v_j)$ forms a monochromatic
  edge. If yes, add $(v_i,v_j)$ to $S$;
 
 \item Set $d^-_G(v_j)$ equals to the number of neighbors that $v_j$ has in $\{v_1,\ldots,v_{j-1}\}$. 
 
 Get $d_G(v_j)$ from the degree oracle and compute $d_G^+(v_j)$.
Set $\mbox{\sc cnt}=\mbox{\sc cnt}+d_G^+(v_j)$.
 Then, depending on whether $\mbox{{\sc cnt}} \leq \tau$, our algorithm performs the following steps.

\If{($\mbox{{\sc cnt}} \leq \tau$)}
{
\begin{itemize}
\item[(i)] Add $(v_j,\ell), \ell \in \left[d_G^+(v_j)\right]$, with probability $\frac{3000 \log n}{\vareps^3 T}$ to $Y$, independently;  
\item[(ii)] Add $v_j$ to $A$ (with its color stored) if at least one $(v_j,\ell)$ is added to $Y$. 
\end{itemize}
}
\ElseIf{($\mbox{{\sc cnt}} > \tau$)}
{
\begin{itemize}
\item[(i)] Add $v_j$ to $B$ (along with the color of $v_j$);
\item[(ii)] For each $v_i \in B$, check whether $(v_i,v_j)$ forms a monochromatic edge. If yes, 

increment $C_{>\tau}$ by $1$. 
\end{itemize}
} 
\end{itemize}

}
If $|S| \leq \frac{60 \log n}{\vareps^2}$, then set $C_{\leq \tau}=0$. Otherwise, set $C_{\leq \tau}=\frac{\vareps^3 T}{3000 \log n}\size{S}$.

Report $\widehat{m}=C_{\leq \tau} +C_{>\tau}$ as the {\sc Output}.
\end{algorithm}
We describe the algorithm and its analysis by breaking the range of $|E|$ into two cases, that is, $T \geq \frac{|E|}{100 |V| \log n}$ (or $|E|\leq 100T|V|\log n = \tau$) and $T<\frac{|E|}{100 |V| \log n}$ (or $|E|> 100T|V|\log n = \tau$). We show that the space complexity of the modified algorithm is $\widetilde{\mathcal{O}}\left(\frac{|E|}{T}\right)$ in the first case and is $\widetilde{\mathcal{O}}(\size{V})$ in the latter case with high probability. Observe that this will imply the desired result as claimed in Theorem~\ref{theo:est-vadgubth}.
\begin{description}
\item[$|E|\leq 100T|V|\log n$:] In this case, by Observation~\ref{obs:cnt}, {\sc cnt} never goes beyond $\tau=100|V|T \log n$. That is, the algorithm behaves exactly same as that of the algorithm presented in Section~\ref{sec:e-known-vadeg} for the case $T > \frac{|E|}{|V|}$. Hence, the algorithm reports the desired output using $\widetilde{\mathcal{O}}\left(\frac{|E|}{|T|}\right)$ space, with high probability. 

\item[$|E|> 100T|V|\log n$:] In this case, by Observation~\ref{obs:cnt}, there will be an instance (say when vertex $v_k$ is exposed) such that {\sc cnt} goes beyond $\tau$ for the first time. Then we start storing all the vertices and their colors in $B=\{v_k,\ldots,v_n\}$. We stop updating $Y$ and $A$ after $v_k$ is exposed. However, we update $S$ until end of the stream as we were doing previously in Section~\ref{sec:e-known-vadeg}. Along with $S$, we maintain the number of monochromatic edges (say $C_{>\tau}$) having both the endpoints in $B=\{{v_k},\ldots,v_n\}$. Note that $C_{>\tau}$ is maintained exactly. Finally, we report $\widehat{m}=C_{\leq \tau}+C_{>\tau}$ as the output, where $0$ or $C_{\leq \tau}=\frac{\vareps^3T}{3000\log n}|S|$ depending on whether $\size{S} \leq \frac{60 \log n}{\vareps^2}$ or not, respectively. By Observation~\ref{obs:degY}, with high probability, $\size{Y}=\widetilde{\mathcal{O}}\left(\size{V}\right)$ for all the instances when {\sc cnt} is less than $\tau$ (that is before the exposure of $v_k$). Also, after the exposure of $v_k$, we are storing all the vertices along with their colors explicitly. So, the space used by the algorithm is $\widetilde{\mathcal{O}}\left({\size{V}}\right)$, with high probability. To see the correctness of the algorithm, let $E_M^B$ be the set of monochromatic edges having both the endpoints in $B=\{v_k.\ldots,v_n\}$. Note that $\size{E_M^B}=C_{>\tau}$. Let $E_M^{V(G)\setminus B}$ be the set of monochromatic edges having at least one vertex in the set $V(G)\setminus B=\{v_1,\ldots, {v_{k-1}}\}$, that is, $E_M^{V(G)\setminus B}=E_M \setminus E_M^B$. Using Chernoff bound arguments (see Lemma~\ref{lem:cher_bound1} in Appendix~\ref{sec:prob}), we have the following lemma. The proof of the following lemma is presented in Appendix~\ref{app:deg-lem}.
\begin{lem}\label{lem:deg}
\begin{itemize}
\item[(i)] If $\size{E_M^{V(G)\setminus B}}\geq \frac{\vareps}{100}T$, then $\frac{\vareps^3T}{3000\log n}|S|$ is a $\left(1\pm \frac{\vareps}{100}\right)$ approximation to $\size{E_M^{V(G)\setminus B}}$ with probability at least $1-\frac{1}{n^{10}}$.
\item[(ii)] If $\size{E_M^{V(G)\setminus B}}\leq \frac{\vareps}{100}T$, $|S|\leq \frac{60 \log n}{\vareps^2} $ with probability at least $1-\frac{1}{n^{10}}$.
\end{itemize} 
\end{lem}
Now let us divide the analysis into two cases, that is, $|S|\geq\frac{60 \log n}{\vareps^2}$ and $|S|<\frac{60 \log n}{\vareps^2}$.
\begin{description}
\item[$|S|\leq\frac{60 \log n}{\vareps^2}$:] In this case, we set $C_{\leq\tau}=0$. So, $\widehat{m}=C_{>\tau}=\size{E_M^B}$ is the output, which is always bounded above by $|E_M|$. By Lemma~\ref{lem:deg} (i), $|S|\leq\frac{60 \log n}{\vareps^2}$ implies $\size{E_M^{V(G)\setminus B}}\leq \frac{\vareps}{25}T$ with probability at least $1-\frac{1}{n^{10}}$. Note that $|E_M|=\size{E_M^{V(G)\setminus B}}+\size{E_M^B}$ and $|E_M|\geq T$. Putting everything together, $\widehat{m}=C_{\leq \tau}+C_{>\tau}$ lies between $\left(1-\frac{\vareps}{25}\right)|E_M|$ and $|E_M|$, with probability at least $1-\frac{1}{n^{10}}$.
\item[$|S|>\frac{60 \log n}{\vareps^2}$:] In this case, we set $C_{\leq \tau}=\frac{\vareps^3T}{3000\log n}|S|$. By Lemma~\ref{lem:deg} (ii), $|S|>\frac{60 \log n}{\vareps^2}$ implies $\size{E_M^{V(G)\setminus B}} {>} \frac{\vareps}{100}T$ with probability at least $1-\frac{1}{n^{10}}$. Also, by Lemma~\ref{lem:deg} (i), $\size{E_M^{V(G)\setminus B}}> \frac{\vareps}{100}|T|$ implies $C_{\leq \tau}=\frac{\vareps^3T}{3000\log n}|S|$ is a $\left(1\pm \frac{\vareps}{100}\right)$ approximation to $\size{E_M^{V(G)\setminus B}}$ with probability at least $1-\frac{1}{n^{10}}$. Combining it with the fact that $C_{> \tau}=\size{E_M^B}$, we have $\widehat{m}=C_{\leq \tau}+C_{>\tau}$ is an $(1\pm \vareps)$-approximation to $|E_M|$, with probability at least $1-\frac{2}{n^{10}}$. 

\end{description}
This finishes the proof for the case $|E|> 100T|V|\log n$.
\end{description}
We have proved the correctness of Algorithm~\ref{algo:random-vadeg} by considering the cases $|E|\leq 100T|V|\log n$ and $|E|> 100T|V|\log n$ separately. We have also shown that the space complexity of Algorithm~\ref{algo:random-vadeg} is $\widetilde{\mathcal{O}}\left(\frac{|E|}{T}\right)$ in the former case and is $\widetilde{\mathcal{O}}(\size{V})$ in the latter case with high probability. Hence, we are done with the proof of Theorem~\ref{theo:est-vadgubth}.

\section{\conflictest and \colorverify in \varand model}\label{sec:est-vardubsc}
\noindent
In this Section, mainly, we show that the power of randomness can be used to design a better solution for the \conflictest problem in the \varand model. The \conflictest problem is the main highlight of our work. We feel that the crucial use of randomness in the input that is used to estimate a substructure (here, monochromatic edges) in a graph, will be of independent interest.

 In this variant, we are given an $ \vareps \in (0,1)$ and a promised lower bound $T$ on $\size{E_M}$, the number of monochromatic edges in $G$, as input and our objective is to determine a $(1 \pm \vareps)$-approximation to $\size{E_M}$. 
 
 \begin{theo} \label{theo:est-vardubth}
 Given any graph $G=(V,E)$ and a coloring function $f: V(G) \rightarrow [C]$ as input in the stream, the \conflictest problem in the \varand model can be solved with high probability in $\widetilde{\mathcal{O}} \left(\frac{\size{V}}{\sqrt T} \right)$ space, where $T$ is a lower bound on the number of monochromatic edges in the graph.\end{theo}
 We prove the above theorem in Section~\ref{sec:subsec-random}. Note that the above algorithm can be used to solve \colorverify in \varand model. In Section~\ref{sec:sep-vaub}, we give a simple algorithm for \conflictest that exploits a \emph{structural} property of the subgraph having only monochromatic edges. However, the space complexity of the algorithm for \colorverify (in Section~\ref{sec:sep-vaub}) is same that of the algorithm for \conflictest (in Section~\ref{sec:subsec-random}). 

\subsection{\conflictest in \varand model (Proof of Theorem~\ref{theo:est-vardubth})}\label{sec:subsec-random}
\subsection*{The proof idea}

\paragraph*{A random sample comes for free -- pick the first few vertices:} Let $v_1,\ldots,v_n$ be the random ordering in which the vertices of $V$ are revealed. Let $R$ be a random subset of 
$\Gamma=\widetilde{\Theta}\left( \frac{n}{\sqrt{T}}\right)$ many vertices of $G$ sampled without replacement~\footnote{$\widetilde{\Theta}(\cdot)$ hides a polynomial factor of $\log n$ and $\frac{1}{\vareps}$ in the upper bound.}. As we are dealing with a random order stream, consider the first $\Gamma$ vertices in the stream; they can be treated as $R$, the random sample. We start by storing all the vertices in $R$ as well as their colors. 
Observe that if the monochromatic degree of any vertex $v_i$ is \emph{large} (say roughly more than $\sqrt{T}$), then it can be well approximated by looking at the number of monochromatic neighbors that $v_i$ has in $R$. As a vertex $v_i$ streams past, there is no way we can figure out its monochromatic degree, unless we store its monochromatic neighbors that appear before it in the stream; if we could, we were done. Our only savior is the stored random subset $R$.

\paragraph*{Classifying the vertices of the random sample $R$ based on its monochromatic degree:}
Our algorithm proceeds by figuring out the influence of the color of $v_i$ on the monochromatic degrees of vertices in $R$.
To estimate this, let $\kappa_{v_i}$ denote the number of monochromatic neighbors that $v_i$ has in $R$. We set a threshold $\tau=\frac{\size{R}}{n}\frac{\sqrt{\vareps T}}{8t}$, where $t=\lceil \log _{1+\frac{\vareps}{10}} n\rceil$. The significance of $t$ will be clear from the discussion below. Any vertex $v_i$ will be classified as a \hmR or \lmR degree vertex depending on its monochromatic degree within $R$, i.e., if $\kappa_{v_i} \geq \tau$, then $v_i$ is a \hmR vertex, else it is a \lmR vertex, respectively. (We use the subscripts m$_R$ to stress the fact that the monochromatic degrees are induced by the set $R$.)
Let $H$ and $L$ be the partition of $V$ into the set of \hmR and \lmR degree vertices in $G$. Let $H_R$ and $L_R$ denote the set of \hmR and \lmR degree vertices in $R$. Notice that, because of the definition of \hmR and \lmR degree vertices, not only the sets $H_R, \, L_R$ are subsets of  $R$, but they are determined by the vertices of $R$ only.

Let $m_h$ and $m_\ell$ denote the sum of the monochromatic degrees of all the \hmR degree vertices and \lmR degree vertices in $G$, respectively. 
So, $m_h=\sum_{v \in H}d_M(v)$ and $m_\ell=\sum_{v \in L}d_M(v)$. Note that $ \widehat{m} = \size{E_M} =\frac{1}{2}\sum_{v \in V} d_M(v) = \frac{1}{2}\left( m_h+m_\ell\right)$.
We will describe how to approximate $m_h$ and $m_\ell$ separately. The formal algorithm is described in Algorithm~\ref{algo:random-order} as {\sc Random-Order-Est}$(\vareps,T)$ that basically executes steps to approximate $m_h$ and $m_\ell$ in parallel.
\begin{algorithm}
\caption{{\sc Random-Order-Est}($\vareps,T)$: \conflictest in \varand model}
\label{algo:random-order}
\KwIn{$G = (V, E)$ and a coloring function $f$ on $V$ in the \varand model, parameters $T$ and $\vareps$.}
\KwOut{$\widehat{m}$, that is, a $(1\pm \vareps)$ approximation to $\size{E_M}$.}
\begin{itemize}
\item $\Gamma=\widetilde{\Theta} \left( \frac{n}{\sqrt{T}}\right)$; $v_1,\ldots,v_n$ be the random ordering in which vertices are revealed and $R=\{v_1,\ldots,v_\Gamma\}$; 
\item  $\kappa_{v_i}, i \in [n],$ denotes the number of monochromatic neighbors of $v_i$ in $R$, 
\item $\widehat{d_{v_i}}, i \in [n],$ denotes the (estimated) monochromatic neighbors of vertices in $G$.
\item $H$ denotes the set of \emph{high} degree vertex in $R$, i.e., $H=\{v_i:\kappa_{v_i} \geq \frac{\size{R}}{n}\frac{\sqrt{\vareps T}}{8t}\}$
and $L=V \setminus H$;  $L_R=L \cap R$ and $H_R=H \cap R$;
\item The vertices in $L$ are partitioned into $t$ buckets as follows:\\
$B_j=\{v_i \in L: \left(1+\frac{\vareps}{10} \right)^{j-1} \leq d_M(v_i) < \left(1+\frac{\vareps}{10}\right)^{j} \}$, where $j \in [t]$.
\end{itemize}
Set $t = \lceil \log _{1+\frac{\vareps}{10}}n \rceil$. If $T <63t^2$, then store all the vertices in $G$ along with their colors. At the end, 
report the exact value of $\size{E_M}$. Otherwise, we proceed through via three building blocks described below and marked as (1),(2), (3) and (4). {Refer to the notations described above this pseudocode.}
\begin{description}
\item[(1)]\underline{{\bf Processing the vertices in $R$, the first $\Gamma$ vertices, in the stream:}}
    
\For{( each vertex $v_i \in R$ exposed in the stream)}
{
Store $v_i$ as well as its color $f(v_i)$.\\
For each edge $(v_{i'},v_i)$ that arrives in the stream, increase the values of $\kappa_{v_{i'}}$ and $\kappa_{v_i}$.\\
}
\item[(2)]\underline{{\bf Computation of some parameters based on vertices in $R$ and their colors:}}

\For{(each $v_i \in R$ with $\kappa_{v_i} \geq  \frac{ \size{R}}{n}\frac{\sqrt{\vareps T}}{8t}$)}
{
Add $v_i$ to $H_R$, and 
set $\widehat{d_{v_i}}=\frac{n}{\size{R}}\kappa_{v_i}$.\\
}
$\widehat{m_h}=\sum\limits_{v_i \in H} \widehat{d_{v_i}}$.\\
Let $L_R = R \setminus H_R$. \\
\For{(each $v_i \in L_R$)}
{
Set $\widehat{d_{v_i}}=\kappa_{v_i}$.\\
}
\item[(3)]\underline{{\bf Processing the vertices in $V(G)\setminus R$ in the stream:}}

\For{(each vertex $v_i \notin R$ exposed in the stream)}
{
Determine the value of $\kappa_{v_i}$. If $\kappa_{v_i} \geq  \frac{ \size{R}}{n}\frac{\sqrt{\vareps T}}{8t}$, find $\widehat{d_{v_i}}=\frac{n}{\size{R}}\kappa_{v_i}$ and add $\widehat{d_{v_i}}$ to the current $\widehat{m_h}.$ \\
Also, for each $v_{i'} \in L_R$, increase the value of $\widehat{d_{v_{i'}}}$ if $(v_{i'},v_i)$ is an edge.\\
}
\item[(4)]\underline{{\bf Post processing, after the stream ends, to return the output:}}

From the values of $\widehat{d_{v_i}}$ for all $v_i \in L_R$, determine the buckets for each vertex in $L_R$. Also, for each $j \in [t]$, find  $ \size{A_j}=\size{L_R \cap B_j}$. Then determine $$\widehat{m_\ell}=\frac{n}{\size{R}}\sum\limits_{ j \in [t]} \size{A_j} \left(1+\frac{\vareps}{10}\right)^j .$$\\
Report $\widehat{m}=\frac{\widehat{m_h}+\widehat{m_\ell}}{2}$ as the final {\sc Output}.
\end{description}
\end{algorithm}
\paragraph*{To approximate $m_h$, the random sample $R$ comes to rescue:} 
 We can find $\widehat{m_h}$, that is, a $\left( 1\pm \frac{\vareps}{10}\right)$ approximation of $m_h$ as described below. For each vertex $v_i \in R$ and each monochromatic edge $(u,v_i)$, $u \in R$, we see in the stream, we increase the value of $\kappa_u$ for $u$ and $\kappa_{v_i}$ for $v_i$. After all the vertices in $R$ are revealed, we can determine $H_R$ by checking whether $\kappa_{v_i} \geq \tau$ for each $v_i \in R$. For each vertex $v_i \in H_R$, we set its approximate monochromatic degree $\widehat{d_{v_i}}$ to be $\frac{n}{\size{R}}\kappa_{v_i}$. We initialize the estimated sum of the monochromatic degree of high degree vertices as $\widehat{m_h}=\sum_{v_i \in H_R}\widehat{d_{v_i}}$. For each vertex $v_i \notin R$ in the stream, we can determine $\kappa_{v_i}$, as we have stored all the vertices in $R$ along with their colors, and hence we can also determine whether $v_i$ is a \hmR degree vertex in $G$. If $v_i \notin R$ is a \hmR degree vertex, we determine $\widehat{d_{v_i}}=\frac{n}{\size{R}}\kappa_{v_i}$ and update $\widehat{m_h}$ by $\widehat{m_h}+\widehat{d_{v_i}}$. Observe that, at the end, $\widehat{m_h}$ is $\sum_{v_i \in H}\widehat{d_{v_i}}$. Recall that $H$ is the set of all \hmR degree vertices in $G$. For each $v_i \in H$, we will show, as in Claim~\ref{clm:random-high}, that $\widehat{d_{v_i}}$ is a $\left(1\pm \frac{\vareps}{10}\right)$-approximation to $d_M(v_i)$ with high probability. This implies that 
 
\begin{equation}
\label{eqn:high-over}
 \left(1-\frac{\vareps}{10}\right)m_h \leq \widehat{m_h} \leq \left(1+\frac{\vareps}{10}\right)m_h
\end{equation}

 \paragraph*{To approximate $m_\ell$, group the vertices in $L$ based on similar monochromatic degree:}
 Recall that $m_\ell=\sum_{v_i \in L}d_M(v_i)$. Unlike the \hmR degree vertices, it is not possible to approximate the {monochromatic} degree of $v_i \in L $ from $\kappa_{v_i}$. To cope up with this problem, we partition the vertices of $L$ into $t$ buckets $B_1,\ldots,B_t$ such that all the vertices present in a bucket have \emph{similar} {monochromatic} degrees, where $t=\lceil \log _{1+\frac{\vareps}{10}} n\rceil$. The bucket $B_j$ is defined as follows: $B_j=\{v_i\in L: \left(1+\frac{\vareps}{10} \right)^{j-1} \leq d_M(v_i) < \left(1+\frac{\vareps}{10}\right)^{j} \}$.

 
 Note that our algorithm will not find the buckets explicitly. It will be used for the analysis only. Observe that 
 $\sum_{j \in [t]}\size{B_j}\left( 1+\frac{\vareps}{10}\right)^{j-1} \leq  m_\ell  < \sum_{j \in [t]}\size{B_j}\left( 1+\frac{\vareps}{10}\right)^{j}$. 
  We can surely approximate $m_\ell$ by approximating $\size{B_j}$s suitably. We estimate $\size{B_j}$s as follows.\remove{ However, we will be able to have a $\left(1+\frac{\vareps}{10}\right)$-approximation of $\size{B_j}$, with high probability, if $\size{B_j}\geq \frac{\sqrt{\vareps T}}{10t}$ as follows.} After the stream of the vertices in $R$ has gone past, we have the set of \lmR degree vertices $L_R$ in $R$ and $\widehat{d_{v_i}}=\kappa_{v_i}$ for each $v_i \in L_R$. For 
  each $v_i \notin R$ in the stream, we determine the monochromatic neighbors of $ v_i$ in $L_R$. It is possible as we have stored all the 
  vertices in $R$ and their colors. For each monochromatic neighbor $v_{i'}\in L_R$ of $v_i$, we increase the value of $\widehat{d}_{v_{i'}}$ 
  of $v_{i'}$.  Observe that, at the end of the stream, $\widehat{d_{v_{i'}}}=d_M(v_{i'})$ for each $v_{i'} \in L_R$, i.e., we can accurately estimate the monochromatic degree of each $v_{i'} \in L_R$. So, we can determine the 
  bucket where each vertex in $L_R$ belongs. Let $A_j$ $(=L_R \cap B_j)$ be the bucket $B_j$ projected onto $L_R$ in the random sample; {note that as $B_j \subseteq L$ and $L_R = L \cap R$, $A_j = R \cap B_j$ also}.  We determine $\widehat{m_\ell}=\frac{n}{\size{R}}\sum_{j \in [t]}\size{A_j}\left(1+\frac{\vareps}{10}\right)^j$. We can show that $\frac{n}{\size{R}}\size{A_j}$ is a $\left(1+\frac{\vareps}{10}\right)$-approximation of $\size{B_j}$, with high probability, if $\size{B_j}\geq \frac{\sqrt{\vareps T}}{10t}$. Also, we can show that, if $\size{B_j}< \frac{\sqrt{\vareps T}}{10t}$, then 
  $\size{A_j} \leq \frac{\size{R}}{n}\frac{\sqrt{\vareps T}}{8t}$ with high probability. Now using the fact that we consider bucketing of only \lmR degree vertices ($L_R$), we can show that 
  \begin{center}
  \begin{equation}
  \label{eqn:low-over}
   \left(1-\frac{\vareps}{10} \right) \left( m_\ell - \frac{\vareps T}{63t}\right) \leq \widehat{m_\ell} \leq \left(1+\frac{\vareps}{10}\right)^2\left(m_\ell + \frac{\vareps T}{56t} \right).
 \end{equation}
\end{center}
Note that $\vareps \in (0,1)$ and $t=\lceil \log _{1+\frac{\vareps}{10}} n \rceil$. Assuming $T \geq 63t^2$,  Equations~\ref{eqn:high-over} and~\ref{eqn:low-over} imply that $\widehat{m}=\frac{1}{2}(\widehat{m_h} +\widehat{m_\ell})$ is a $(1 \pm \vareps)$-approximation to $\size{E_M}$. If $T < 63t^2$, then note that $n=\widetilde{{\cal O}}\left(\frac{n}{\sqrt{T}}\right)$. So, in that case, we store all the vertices along with their colors and compute the exact value of $\size{E_M}$. 

\subsection*{Proof of correctness}
The correctness of the algorithm follows trivially if $T < 63t^2$. So, let us assume that $T \geq 63t^2$.
In the \varand model, we consider the first $\widetilde{\Theta}\left( \frac{n}{\sqrt{T}}\right)$ vertices as the random sample $R$ without replacement. Using the Chernoff bound for sampling without replacement (See Lemma~\ref{lem:without-replace} in Appendix~\ref{sec:prob}), we can have the following lemma ({The proof is in Appendix~\ref{sec:deviation}}), which will be useful for the correctness proof of Algorithm~\ref{algo:random-order} ({\sc Random-Order-Est}($\vareps,T)$) in case of $T \geq 63t^2$.
\begin{lem}\label{lem:cher-random-algo}
\begin{itemize}
\item[(i)] For each $j \in [t]$ with $\size{B_j} \geq \frac{\sqrt{\vareps T}}{10t}$, $\pr \left( \size{{\size{B_j \cap R}}-
\frac{\size{R}\size{B_j}}{n}} \geq \frac{\vareps}{10}\frac{\size{R}\size{B_j}}{n}  \right) \leq \frac{1}{n^{10}}$.
\item[(ii)] For each $j \in [t]$ with $\size{B_j} < \frac{\sqrt{\vareps T}}{10t}$, $\mathbb{P} \left( {{\size{B_j \cap R}}} \geq \frac{\size{R}}{n}\frac{\sqrt{\vareps T}}{8t}  \right) \leq \frac{1}{n^{10}}$.
\item[(iii)] For each vertex $v_i$ with $d_M(v_i) \geq \frac{\sqrt{\vareps T}}{10t}$, $\mathbb{P} \left( \size{\kappa_{v_i} - \frac{\size{R}d_M(v_i)}{n}} \geq \frac{\vareps}{10}\frac{\size{R}d_M(v_i)}{n}\right) \leq \frac{1}{n^{10}}$.
\item[(iv)] For each vertex $v_i$ with $d_M(v_i) < \frac{\sqrt{\vareps T}}{10t}$, $\mathbb{P} \left( {\kappa_{v_i}} \geq \frac{ \size{R}}{n}\frac{\sqrt{\vareps T}}{8t}\right) \leq \frac{1}{n^{10}}$.
\end{itemize}
\end{lem}

\remove{ Recall that $\widehat{m_h}=\sum\limits_{v_i:\kappa_i \geq \frac{\size{R}}{n}\frac{\sqrt{\vareps T}}{8t}}\widehat{d_i}$ and $\widehat{m_\ell}=\frac{n}{\size{R}}\sum\limits_{j\in [t]}\size{A_j}\left(1+\frac{\vareps}{10}\right)^{j}$.\\
 Also, $m_h=\sum\limits_{v_i:\kappa_i \geq \frac{\size{R}}{n}\frac{\sqrt{\vareps T}}{8t}} d_M(v_i)$ and $m_\ell=\sum\limits_{v_i:\kappa_i < \frac{\size{R}}{n}\frac{\sqrt{\vareps T}}{8t}} d_M(v_i)$. \\
So, $\size{E_M}=\frac{1}{2}(m_h+m_\ell)$. }
The correctness proof of the algorithm is divided into the following two claims.
\begin{cl}\label{clm:random-high}
$\left(1-\frac{\vareps}{10}\right)m_h\leq \widehat{m_h} \leq \left(1+\frac{\vareps}{10}\right)m_h $ with probability at least $1-\frac{1}{n^{9}}$.
\end{cl}
\begin{cl}\label{clm:random-low}
$\left(1-\frac{\vareps}{10}\right)\left(m_\ell-\frac{\vareps T}{63t}\right) \leq \widehat{m_\ell} \leq \left(1+\frac{\vareps}{10}\right)^2\left(m_\ell+\frac{\vareps T}{56t}\right) $ with probability at least $1-\frac{1}{n^{7}}$. 
\end{cl}
Assuming the above two claims hold and taking $\vareps \in (0,1)$, $t=\lceil \log _{1+\frac{\vareps}{10}} n \rceil$ and $T \geq 63t^2$, observe  that $\widehat{m}=\frac{1}{2}(\widehat{m_h}+\widehat{m_\ell})$ is a $(1\pm \vareps)$ approximation of $\size{E_M}=m_h+m_\ell$ with high probability.
Thus, it remains to prove Claims~\ref{clm:random-high} and~\ref{clm:random-low}.

\begin{proof}[Proof of Claim~\ref{clm:random-high}]
Note that $m_h=\sum\limits_{v_i:\kappa_{v_i} \geq \frac{\size{R}}{n}\frac{\sqrt{\vareps T}}{8t}} d_M(v_i)$ and $\widehat{m_h}=\sum\limits_{v_i:\kappa_{v_i} \geq \frac{\size{R}}{n}\frac{\sqrt{\vareps T}}{8t}}\widehat{d_{v_i}}$. 

From Lemma~\ref{lem:cher-random-algo}~(iv) and~(iii), $\kappa_{v_i} \geq \frac{\size{R}}{n}\frac{\sqrt{\vareps T}}{8t}$ implies that $\widehat{d_{v_i}}$ is an $\left(1\pm \frac{\vareps}{10}\right)$ approximation to $d_M(v_i)$ with probability at least $1-\frac{2}{n^{10}}$. Hence, we have $\left(1-\frac{\vareps}{10}\right)m_h\leq \widehat{m_h} \leq \left(1+\frac{\vareps}{10}\right)m_h $ with probability at least $1-\frac{1}{n^9}$.
 \end{proof}
 \begin{proof}[Proof of Claim~\ref{clm:random-low}]
 Note that $m_\ell= {\sum_{v_i \in L} d_M(v_i) =} \sum\limits_{v_i:\kappa_{v_i} < \frac{\size{R}}{n}\frac{\sqrt{\vareps T}}{8t}} d_M(v_i)$ and $\widehat{m_\ell}=\frac{n}{\size{R}}\sum\limits_{j\in [t]}\size{A_j}\left(1+\frac{\vareps}{10}\right)^{j}$.
 Recall that the vertices in $L$ are partitioned into $t$ buckets as follows:\\
$B_j=\{v_i \in L: \left(1+\frac{\vareps}{10} \right)^{j-1} \leq d_M(v_i) < \left(1+\frac{\vareps}{10}\right)^{j} \}$, where $j \in [t]$.
By Lemma~\ref{lem:cher-random-algo} {(iv)}, $\kappa_{v_i} < \frac{\size{R}}{n}\frac{\sqrt{\vareps T}}{8 t}$ implies that $d_M(v_i) \leq \frac{\sqrt{\vareps T}}{7 t}$ with probability $1-\frac{1}{n^{10}}$. So, we have the following observation. 
\begin{obs}\label{obs:low-deg-low-buck}
Let $j \in [t]$ be such that $\size{A_j} \neq 0$ ($\size{B_j}\neq 0$). Then, with probability at least $1-\frac{1}{n^{10}}$, the {monochromatic} degree of each vertex in $A_j$ as well as $B_j$ is at most $\frac{\sqrt{\vareps T}}{7 t}$, that is, $\left( 1+\frac{\vareps}{10}\right)^{j} \leq \frac{\sqrt{\vareps T}}{7 t}$.
\end{obs}

To upper and lower bound $\widehat{m_\ell}$ in terms of $m_\ell$, we upper and lower bound $m_\ell$ in terms of $\size{B_j}$'s as follows; {for the upper bound, we break the sum into two parts corresponding to large and small sized buckets}:
\begin{eqnarray*}\label{eqn:exact-low}
&& \sum\limits_{j \in [t]}\size{B_j}\left( 1+\frac{\vareps}{10}\right)^{j-1} \leq  m_\ell  < \sum\limits_{j \in [t]}\size{B_j}\left( 1+\frac{\vareps}{10}\right)^{j} \\
&& \sum\limits_{j \in [t]}\size{B_j}\left( 1+\frac{\vareps}{10}\right)^{j-1} \leq  m_\ell  < \sum\limits_{j \in [t]:\size{B_j}\geq \frac{\sqrt{\vareps T}}{9t}}\size{B_j}\left( 1+\frac{\vareps}{10}\right)^{j} +  \sum\limits_{j \in [t]:\size{B_j}< \frac{\sqrt{\vareps T}}{9t}}\size{B_j}\left( 1+\frac{\vareps}{10}\right)^{j}
 \end{eqnarray*} 
  By Observation~\ref{obs:low-deg-low-buck}, we bound $m_\ell$ in terms of $\size{B_j}$'s with probability $1-\frac{1}{n^9}$.
  \begin{eqnarray*}
 && \sum\limits_{j \in [t]}\size{B_j}\left( 1+\frac{\vareps}{10}\right)^{j-1} \leq  m_\ell  < \sum\limits_{j \in [t]:\size{B_j}\geq \frac{\sqrt{\vareps T}}{9t}}\size{B_j}\left( 1+\frac{\vareps}{10}\right)^{j} +  t \cdot \frac{\sqrt{\vareps T}}{9t}\frac{\sqrt{\vareps T}}{7t}
  \end{eqnarray*}
  This implies the following Observation:
  \begin{obs}\label{obs:exact-low}
$\sum\limits_{j \in [t]}\size{B_j}\left( 1+\frac{\vareps}{10}\right)^{j-1} \leq  m_\ell  < \sum\limits_{j \in [t]:\size{B_j}\geq \frac{\sqrt{\vareps T}}{9t}}\size{B_j}\left( 1+\frac{\vareps}{10}\right)^{j} +  \frac{\vareps T}{63 t}$ holds with probability at least $1-\frac{1}{n^9}$. 
\end{obs}

Now, we have all the ingredients to show that $\widehat{m_\ell}$ is a $(1 \pm \vareps)$ approximation of $m_\ell$. {To get to $\widehat{m_\ell}$, we need to focus on \lmR vertices of $R$, i.e., $A_j$'s.} Breaking $ \widehat{m_\ell} = \frac{n}{\size{R}}\sum\limits_{j\in [t]}\size{A_j}\left(1+\frac{\vareps}{10}\right)^{j}$ depending on {small and large} values of $\size{A_j}$'s (recall $A_j=L_R \cap B_j = R \cap B_j$), we  have
 \begin{eqnarray}\label{eqn:low-random}
 \widehat{m_\ell}&=& \frac{n}{\size{R}} \left[\sum\limits_{j\in [t]:\size{A_j} \geq \frac{\size{R}}{n}\frac{\sqrt{\vareps T}}{8 t}}\size{A_j}
 \left(1+\frac{\vareps}{10}\right)^{j} + \sum\limits_{j\in [t]:\size{A_j} < \frac{\size{R}}{n}\frac{\sqrt{\vareps T}}{8 t}}\size{A_j}
 \left(1+\frac{\vareps}{10}\right)^{j}\right]
 \end{eqnarray}
  Note that $A_j=B_j \cap R$. By Lemma~\ref{lem:cher-random-algo}~(ii), $\size{A_j} \geq \frac{\size{R}}{n}\frac{\sqrt{\vareps T}}{8 t}$ implies $\size{B_j}\geq \frac{\sqrt{\vareps T}}{10 t}$ with probability at least $1-\frac{1}{n^{10}}$. Also, applying Lemma~\ref{lem:cher-random-algo}~(i), $\size{B_j}\geq \frac{\sqrt{\vareps T}}{10 t}$ implies $\size{A_j}$ is an $\left(1\pm\frac{\vareps}{10}\right)$-approximation to $\frac{\size{R}\size{B_j}}{n}$ with probability at least $1-\frac{1}{n^{10}}$. So, we have the following observation. 
 \begin{obs}\label{obs:lemma(i)&(ii)}
 Let $j \in [t]$ be such that $\size{A_j} \geq \frac{\size{R}}{n}\frac{\sqrt{\vareps T}}{8 t}$. Then $\size{A_j}$ is an $\left(1\pm\frac{\vareps}{10}\right)$-approximation to $\frac{\size{R}\size{B_j}}{n}$ with probability at least $1-\frac{2}{n^{10}}$, that is, $\frac{n}{\size{R}}\size{A_j}$ is an $\left(1\pm\frac{\vareps}{10}\right)$-approximation to $\size{B_j}$ with probability at least $1-\frac{2}{n^{10}}$
 \end{obs}
 
 \remove{From Lemma~\ref{lem:cher-random-algo}~(i), $\size{A_j} < \frac{\size{R}}{n}\frac{\sqrt{\vareps T}}{8 t}$ implies that $\size{B_j'}\leq \frac{\sqrt{\vareps T}}{7t}$ with probability at least $1-\frac{1}{n^{10}}$.} 
 By the above observation along with Equation~\ref{eqn:low-random}, we have the following upper bound on $\widehat{m_\ell}$ with probability at least $1-\frac{1}{n^9}$.

 \begin{eqnarray*}
 \widehat{m_\ell} &\leq& {
 \sum\limits_{j\in [t]:\size{A_j} \geq \frac{\size{R}}{n}\frac{\sqrt{\vareps T}}{8 t}}\left(1+\frac{\vareps}{10}\right)\size{B_j}\left(1+\frac{\vareps}{10}\right)^{j} +
  \sum\limits_{j\in [t]:\size{A_j} < \frac{\size{R}}{n}\frac{\sqrt{\vareps T}}{8 t}}\frac{n}{\size{R}}\size{A_j}\left(1+\frac{\vareps}{10}\right)^{j}}
  \\
  &\leq & \left(1+\frac{\vareps}{10}\right)^2 \left[ \sum\limits_{j\in [t]:\size{A_j} \geq \frac{\size{R}}{n}\frac{\sqrt{\vareps T}}{8 t}}\size{B_j}\left(1+\frac{\vareps}{10}\right)^{j-1} +   \sum\limits_{j\in [t]:\size{A_j} < \frac{\size{R}}{n}\frac{\sqrt{\vareps T}}{8 t}}\frac{\sqrt{\vareps T}}{8t}\left(1+\frac{\vareps}{10}\right)^{j-2}
 \right]
 \end{eqnarray*}
 Now by Observations~\ref{obs:exact-low} and ~\ref{obs:low-deg-low-buck}, we have the following with probability at least $1-\frac{1}{n^8}$.
 \begin{eqnarray*}
 \widehat{m_\ell} &\leq& \left(1+\frac{\vareps}{10}\right)^2 \left( m_\ell + t \cdot \frac{\sqrt{\vareps T}}{8t}\frac{\sqrt{\vareps T}}{7t}\right) \\
 &=&  \left(1+\frac{\vareps}{10}\right)^2 \left( m_\ell + \frac{\vareps T}{56 t}\right)
 \end{eqnarray*}
 Now, we will lower bound $\widehat{m_\ell}$. From Equation~\ref{eqn:exact-low}, we have 
 \begin{eqnarray*}
 \widehat{m_\ell} &\geq& \frac{n}{\size{R}} \sum\limits_{j\in [t]:\size{A_j} \geq \frac{\size{R}}{n}\frac{\sqrt{\vareps T}}{8 t}}\size{A_j}
 \left(1+\frac{\vareps}{10}\right)^{j} 
 \end{eqnarray*}
By Observation~\ref{obs:lemma(i)&(ii)},  $\size{A_j} \geq \frac{\size{R}}{n}\frac{\sqrt{\vareps T}}{8 t}$ implies $\frac{n}{\size{R}}\size{A_j}$ is an $\left(1\pm\frac{\vareps}{10}\right)$-approximation to $\size{B_j}$ with probability at least $1-\frac{2}{n^{10}}$. So, the following lower bound on $\widehat{m_\ell}$ holds with probability at least $1-\frac{1}{n^{9}}$.
   \begin{eqnarray*}
 \widehat{m_\ell} &\geq& \left(1-\frac{\vareps}{10} \right) \sum\limits_{j\in [t]:\size{A_j} \geq \frac{\size{R}}{n}\frac{\sqrt{\vareps T}}{8 t}}\size{B_j}
 \left(1+\frac{\vareps}{10}\right)^{j}
 \end{eqnarray*}
By Lemma~\ref{lem:cher-random-algo}~(i), if $\size{B_j} \geq \frac{\sqrt{\vareps T}}{9t}$, then $\size{A_j} \geq \frac{\sqrt{\vareps T}}{8t}$ with probability at least $1-\frac{1}{n^{10}}$. Hence, we have the following lower bound on $m_\ell$ with probability at least $1-\frac{1}{n^8}$.
 \begin{eqnarray*}
\widehat{m_\ell} &\geq& \left(1-\frac{\vareps}{10} \right) \sum\limits_{j\in [t]:\size{B_j} \geq {\frac{\sqrt{\vareps T}}{9 t}}}\size{B_j}
 \left(1+\frac{\vareps}{10}\right)^{j} 
 \end{eqnarray*}
 Now by Observation~\ref{obs:exact-low}, we have the following with high probability at least $1-\frac{1}{n^7}$.
 \begin{eqnarray*}
 \widehat{m_\ell} &\geq &  \left(1-\frac{\vareps}{10} \right) \left( m_\ell - \frac{\vareps T}{63t}\right).
 \end{eqnarray*}

\end{proof}

\subsection{\colorverify in \varand model} \label{sec:sep-vaub}
\noindent
Using a structural property of the graph, we design a simple algorithm to solve the \colorverify problem in the \varand model.

\begin{theo}\label{thm:sep-vardubth} 
Given any graph $G=(V,E)$ and a coloring function $f:V(G) \rightarrow [C]$ and a parameter $\vareps>0$ as input, there exists an algorithm that solves the \colorverify problem in the \varand streaming model using space $ \widetilde{\mathcal{O}} \left(\frac{\size{V}}{\sqrt{\vareps \size{E}}} \right)$ with high probability. \end{theo}

Let $G'$ denote the subgraph of $G$ consisting of only monochromatic edges in $G$. The lemma stated below guarantees that either there exists a large matching of size at least $\sqrt{\vareps m}$ in $G'$ or there exists a vertex of degree at least $\sqrt{\vareps m}$ in $G'$.

\begin{lem}[\cite{DBLP:series/txtcs/Jukna11}]\label{lem:sunflower} Let $G=(V,E)$ be a graph and $f:V(G) \rightarrow [C]$ be a coloring function such that at least $\vareps$ fraction of the edges of $E(G)$ are known to be monochromatic. Then, either there is a matching of size at least $\sqrt{\vareps m}$ or there exists a vertex of degree at least $\sqrt{\vareps m}$ in the subgraph $G'$ defined on the monochromatic edges of $G$. \end{lem}
\begin{algorithm}[H]
\SetAlgoLined
\KwIn{$G = (V, E)$ and a coloring function $f$ on $V$ in the \varand model}
\KwOut{The algorithm verifies if $f$ is $\varepsilon$-far from valid or not}

{}{Let } $S$ {}{ be the set of stored vertices and their colors. Initially, } $S$ {}{ is empty.}
    \For{$i \gets 1$ {}{ to } $|V|$} { 
       {}{let } $u$ {}{ be the } $i^{th}$ {}{ vertex of the stream} \\
       {}{Store } $u$ {}{ and its color } $f(u)$ {}{ in } $S$ {}{ with probability } $\mathcal{O}\left(\frac{\log n}{\sqrt m} \right)$ \\
       \For{{}{every vertex } $v$ {}{ in } $S$} {
           {}{Check if } $(v, u)$ {}{ is an edge and } $f(v) = f(u)$
       }
    }  
{}{Output } $f$ {}{ is valid if none of the edges sampled are conflicting, else output that } $f$ {}{ is } $\varepsilon$ {}{-far from being valid.}
\vspace{1 cm}
\caption{Algorithm: \colorverify in \varandlong model}
\label{algo:random-order-verify}
\end{algorithm}
The algorithm is as simple as it can get. We sample independently and uniformly at random the vertices in stream with probability $p=\min\{1,\frac{10\log n}{\sqrt{m}}\}$~\footnote{{For simplicity of presentation, we assumed that, the number of edges $m$ in graph $G$ is known before the stream starts. However, this assumption can be removed by a simple tweak of starting with a value of $m$ and increasing it in stages and adjusting the random sample accordingly. This is common in streaming algorithms.}} and store these vertices along with their colors. Let $S\subseteq V$ be the set of sampled vertices. When a vertex appears in a stream, we check if it forms a monochromatic edge with one of the stored vertices in $S$. At the end of the stream, the algorithm declares the graph to be properly colored {(valid)} if it can not find a monochromatic edge, else it declares the instance to be $\vareps$-far from being monochromatic.

We show that Theorem \ref{thm:sep-vardubth} follows easily using Lemma \ref{lem:sunflower}.
\begin{proof} We consider the following two cases.
\begin{itemize}
	\item {Case 1 -- There exists a matching of size at least $\sqrt{\vareps m}$:} Note that all these matched edges are monochromatic. Let $(u,v)$ denote an arbitrary matched edge where $u$ appears in the stream before $v$. Now, the edge $(u,v)$ will be detected as monochromatic if vertex $u$ has been sampled by the algorithm. The probability that vertex $u$ is sampled is $\frac{10\log n}{\sqrt{m}}$. Since, there are $\sqrt{\vareps m}$ matched monochromatic edges, the algorithm will detect at least one of these matched monochromatic edges with probability at least $(1-1/n^2)$.


	\item {Case 2 -- There exists a vertex of degree at least $\sqrt{\vareps m}$:} In this case most of the monochromatic edges may be incident on very few high degree vertices. To detect these edges, we want to store either the high degree vertices or one of its neighbours. But, if these high degree vertices appear at the beginning of the stream and we fail to sample them, then we may not detect a monochromatic edge. This is where the \emph{random order} of vertices arriving in the stream comes into play. Now, assuming random order of vertices in the stream, at least $\frac{1}{5} \sqrt{\vareps m}$ neighbors of $v$ should appear before $v$ in the stream with probability at least $(1-e^{-\frac{9}{50} \sqrt{\vareps m}})$ \remove{\footnote{\comments{If $m$ is small, success probability may be small as well}}}. Since we sample every vertex with probability $\frac{10\log n}{\sqrt{m}}$, with high probability at least $(1-1/n^2)$ one of its neighbors will be stored.


\end{itemize}
\end{proof}

\section{Lower bound for \conflictest in \varand model}\label{sec:lower-rand}
\noindent
In this Section, we show a lower bound of $\Omega\left(\frac{n}{T^2}\right)$ for \conflictest in 
\varandlong via a reduction from a variation of {\sc Multiparty Set Disjointness} problem called $\mbox{{\sc Disjointness}}_R(t,n,p)$, played among $p$ players: 
\remove{Consider a $V$ of order $t\times n$ matrix having $t$ (row) vectors $M_1,\ldots,M_t \in \{0,1\}^n$ such that each 
entry of matrix $M$ is given to one of $p$ players chosen uniformly at random.}Consider a matrix of order $t\times n$ having $t$ (rows) vectors $M_1,\ldots,M_t \in \{0,1\}^n$ such that each 
entry of matrix $M$ is given to one of the $p$ players chosen uniformly at random. The objective is 
to determine whether there exists a column where all the entries are $1$s. If $t\geq 2$ and $p = \Omega(t^2)$, Chakrabarti et al. showed that any randomized protocol requires $\Omega\left(\frac{n}{t}\right)$ bits of communication~\cite{DBLP:journals/toc/ChakrabartiC016}. They showed that the lower bound holds under a promise called the {\sc unique intersection promise} which states that there exists at most a single column where all the entries are $1$s and every other column of the matrix has Hamming weight either $0$ or $1$. \remove{, and every column besides this one has Hamming weight either $0$ or
$1$. The promise is often referred as {\sc unique intersection promise}.} Moreover, the lower bound holds 
even if all the $p$ players know the random partition of the entries of matrix $M$. 
\begin{theo}\label{thm:rand-lower}
Let $n,T  \in \N$ be such that $4 \leq T \leq {n \choose 2}$. Any constant pass streaming algorithm that takes the vertices and edges of a graph $G(V,E)$ (with $\size{V}=\Theta(n)$ and $\size{E}=\Theta(m)$) and a coloring function $f:V \rightarrow [C]$ in the \varand model, and determines whether the monochromatic edges in $G$ is $0$ or $\Omega(T)$ with probability $2/3$, requires $\Omega\left( \frac{n}{T^2}\right)$ bits of space.
\end{theo}
\begin{proof}
Without loss of generality, assume that $\sqrt{T} \in \N$. Consider the $\mbox{\sc Disjointness}_R
\left(\sqrt{T},\frac{n}{\sqrt{T}},p\right)$ problem with {\sc Unique Intersection promise} when all of the $p$ players know the random partition of the entries of the relevant matrix $M$. Note that $M$ is of order $[\sqrt{T}]\times \left[\frac{n}{\sqrt{T}}\right]$ and $p=AT$ for 
some suitable constant $A \in \N$. Also, consider a graph $G$, with $V(G)=\{v_{ij}:i\in \left[\sqrt{T}\right], j \in \left[\frac{n}{\sqrt{T}}\right] \}$, having $\frac{n}{\sqrt{T}}$ many vertex disjoint cliques such that $\{v_{1j},\ldots,v_{\sqrt{T}j}\}$ forms a clique for each $j\in [n]$, {i.e., a column of $M$ forms a clique. Also, notice that each clique has $\Theta(T)$ edges.} Let us assume that there is an $r$-pass streaming algorithm $\cS$, with space complexity $s$ bits, that solves \conflictest for the above graph $G$ in the \varand model.  
Now, we give a protocol $\cA$ for $\mbox{\sc Disjointness}_R
\left(\sqrt{T},\frac{n}{\sqrt{T}},p\right)$ with communication cost $O(rsp)$. Using the fact that the lower bound of $\mbox{\sc Disjointness}
\left(\sqrt{T},\frac{n}{\sqrt{T}},p\right)$ is $\Omega\left(\frac{n/\sqrt{T}}{\sqrt{T}}\right)$ along with the fact that $p=AT$ and $r$ is a constant, we get $s=\Omega\left(\frac{n}{T^2}\right)$.

\paragraph*{Protocol $\cA$ for $\mbox{\sc Disjointness}_R
\left(\sqrt{T},\frac{n}{\sqrt{T}},p\right)$:}
Let $P_1,\ldots,P_p$ denote the set of $p$ players. For $k\in [p]$, $V_k=\{v_{ij}:M_{ij}~\mbox{is with } P_k\}$, where $M_{ij}$ denotes the element present in the 
$i$-th row and $j$-th column of matrix $M$. Note that there is a one-to-one correspondence between the entries of $M$ and the vertices in $V(G)$. Furthermore, there is a one-to-one correspondence between the columns of matrix $M$ and the cliques in graph $G.$ We assume that all the $p$ players know the graph structure completely as well as both the one-to-one correspondences. The protocol proceeds as follows: for each $k\in [p]$, {player $P_k$ determines a random permutation $\pi_k$ of the vertices in $V_k$}. Also, for each $k\in [p]$, player $P_k$ determines the colors of the vertices in $V_k$ by the following rule: if $M_{ij}=1$, then color vertex $v_{ij}$ with color $C_*$. Otherwise, {for $M_{ij}=0$,} color vertex $v_{ij}$ with color $C_i$. Player $P_1$ initiates the streaming algorithm and it goes over $r$-rounds. 

\begin{description}
\item[Rounds $1$ to $r-1$:] For $k\in [p]$, each player resumes the streaming algorithm by exposing the vertices in $V_k$, along with their colors, in the order dictated by $\pi_k$. Also, $P_k$ adds the respective edges {to previously exposed vertices} when the current vertex is exposed to satisfy the {basic} requirement of \va model. This is possible because all players know the graph $G$ and the random partition of the entries of  matrix $M$ among $p$ players. After exposing all the vertices in $V_k$, as described, $P_k$ sends the current memory state to player $P_{k+1}$. Assume that $P_1=P_{p+1}$.

\item[Round $r$:] All the players behave similarly as in the previous rounds, except that, the player $P_p$
does not send the current memory state to $P_1$. Rather, $P_p$ decides whether there is a column in $M$ with all $1$s if the streaming algorithm $\cS$ decides that there are $\Omega(T)$ many monochromatic edges in $G$. Otherwise, if $\cS$ decides that there is no monochromatic edge in $G$, then $P_p$ decides that all the columns of $M$ have weight either $0$ or $1$. Then $P_p$ sends the output to all other players. 
\end{description}
The vertices of graph $G$ are indeed exposed randomly to the streaming algorithm. It is because the entries of matrix $M$ are randomly partitioned among the players and each player also generates a random permutation of the vertices corresponding to the entries of matrix $M$ available to them.
From the description of the protocol $\cA$, the memory state of the streaming algorithm (of space complexity $s$) is communicated $(r-1)p+(p-1)$ times and $p-1$ bits is communicated at the end by player $P_p$ to broadcast the output. Hence, the communication cost of the protocol $\cA$ is at most $O(rsp)$.

Now we are left to prove the correctness of the protocol $\cA$. If there is a column in $M$ with all $1$s, then all the vertices corresponding to entries of that column are colored with color $C_*$. Recall that there is a one-to-one correspondence between the columns in matrix $M$ and cliques in the graph $G$. So, all the vertices of the clique, corresponding to the column having all $1$s, are colored with the color $C_*$. As the size of each clique in the graph $G$ is $\sqrt{T}$, there are at most $\Omega(T)$ monochromatic edges. To prove the converse, assume that there is no 
column in the matrix $M$ having all $1$s. By {\sc Unique Intersection Promise}, all the columns have hamming weight at most $1$. We will argue that there is no monochromatic edge in $G$. Consider an edge $e$ in $G$. By the structure of $G$, the two vertices of $e$ {must be} in the same clique, {say the $j$-th clique}, that is, let $e=\{v_{i_1j},v_{i_2j}\}$. By the coloring scheme used by the protocols, $v_{i_1j}$ and $v_{i_2j}$ are colored 
according to the values of $M_{i_1j}$ and $M_{i_2j}$, respectively. Note that both  $M_{i_1j}$ and $M_{i_2j}$ belong to $j$-th column. As the hamming weight of every column is at most $1$, there are three possibilities: 
\begin{itemize}
\item[(i)] $M_{i_1j}=M_{i_2j}=0$, that is, $v_{i_1j}$ and $v_{i_2j}$ are colored with color 
$C_{i_1}$ and $C_{i_2}$, respectively; 
\item[(ii)] $ M_{i_1j}=0~\mbox{and}~M_{i_2j}=1$, that is, $v_{i_1j}$ and $v_{i_2j}$ are colored with color 
$C_{i_1}$ and $C_*$, respectively;  
\item[(iii)]  $ M_{i_1j}=1~\mbox{and}~M_{i_2j}=0$, that is, $v_{i_1j}$ and $v_{i_2j}$ are colored with color $C_*$ and $C_{i_2}$, respectively. 
\end{itemize}
In any case, the edge  $e=\{v_{i_1j},v_{i_2j}\}$ is not monochromatic. This establishes the correctness of protocol $\cA$ for $\mbox{\sc Disjointness}_R
\left(\sqrt{T},\frac{n}{\sqrt{T}},p\right)$. 
\end{proof}

\section{Conclusion and Discussion}
\label{sec:conclude}
\noindent
In this paper, we introduced a graph coloring problem to streaming setting with a different flavor -- the coloring function streams along with the graph. We study the problem of \conflictest (estimating the 
number of monochromatic edges) and \colorverify (detecting a separation between the number of valid edges) in \va, \vadeg, and \varand models. Our algorithms for \va and \vadeg are tight upto polylogarithmic factors.  
However, a matching lower bound on the space complexity for \varand model is still elusive. There is 
a gap between our upper and lower bound results for \varand model in terms of the exponent in $T$. Our hunch is that the upper bound is tight.
Specifically, we obtained an upper bound of $\widetilde{{\cal O}}\left(\frac{n}{\sqrt{T}}\right)$ and the lower bound is $\Omega\left(\frac{n}{T^2}
\right)$. Here we would like to note that the lower bound also holds in \al and \vadeg model when the vertices are exposed in a random order. However, we feel that our algorithm for \conflictest in \varand model is tight upto polylogarithmic factors. We leave this problem open.

We feel the \emph{edge coloring} counterpart of the vertex coloring problem proposed in the paper will be worthwhile to study. Let the edges of $G$ be colored with a function $f:E(G) \rightarrow [C]$, for $C \in \N$. A vertex $u \in V(G)$ is said to be a \emph{validly} colored vertex if no two edges incident on $u$ have the same color. An edge coloring is valid if all vertices are validly colored. Consider the \al model for the edge coloring problem. As all edges incident on an exposed vertex $u$ are revealed in the stream, if we can solve a duplicate element finding problem on the colors of the edges incident on $u$, then we are done! It seems at a first glance that all the three models of \va, \al and \ea will be difficult to handle for the edge coloring problem on streams of graph and edge colors. It would be interesting to see if the edge coloring variant of the problems we considered in this paper, admit efficient streaming algorithms. We plan to look at this problem next.

\bibliographystyle{alpha}
\bibliography{lipics-v2019-sample-article}

\newpage
\appendix

\section{Some probability results}\label{sec:prob}
\begin{lem}[\cite{DBLP:books/daglib/0025902}(Chernoff-Hoeffding bound)]
\label{lem:cher_bound1}
Let $X_1, \ldots, X_N$ be independent random variables such that $X_i \in [0,1]$. For $X=\sum\limits_{i=1}^N X_i$ and $\mu=\mathbb{E}[X]$, the following holds for any $0 \leq \delta \leq 1$:
\begin{itemize}
\item[(i)]$ \mathbb{P}(X  \geq (1+\delta)\mu) \leq \exp{\left( \frac{-\mu \delta^{2}}{3}\right)} $;
\item[(ii)]$ \mathbb{P}(X  \leq (1-\delta)\mu) \leq \exp{\left( \frac{-\mu \delta^{2}}{3}\right)} $;
\item[(iii)] Furthermore, if $\mu \leq t$, then the following holds.
$$ \mathbb{P}(X  \geq (1+\epsilon)t) \leq \exp{\left( \frac{-t \delta^{2}}{3}\right)}.$$
\end{itemize}

\end{lem}
\begin{lem}[\cite{DBLP:journals/eatcs/Mulzer18}] \label{lem:without-replace}
Let $I=\{1,\ldots,N\}$, $r \in [N]$ be a given parameter. If we sample a subset $R$ without replacement, then the following holds for any $J \subset I$ and $\delta \in (0,1)$.
\begin{itemize}
\item[(i)] $\pr \left( \size{J \cap R} \geq (1+\delta)\size{J}\frac{r}{N}\right) \leq \exp{\left(-\frac{\delta^2\size{J}r}{3N}\right)}$;
\item[(ii)] $\pr \left( \size{J \cap R} \leq (1-\delta)\size{J}\frac{r}{N}\right) \leq \exp{\left(-\frac{\delta^2\size{J}r}{3N}\right)}$;
\item[(iii)] Further more, we have the following if $\size{J} \leq k$, then the following holds.
 \[\pr \left( \size{J \cap R} \geq (1+\delta)k\frac{r}{N}\right) \leq \exp{\left(-\frac{\delta^2kr}{3N}\right)}\]
\end{itemize}
\end{lem}

\section{Proof of Lemma~\ref{lem:deg}} \label{app:deg-lem}
\begin{lem}[Restatemnet of Lemma~\ref{lem:deg}]
\begin{itemize}
\item[(i)] If $|E_M^a|\geq \frac{\vareps}{100}|T|$, then $\frac{\vareps^3T}{3000\log n}|S|$ is an $\left(1\pm {\vareps}\right)$ approximation to $|E_M^a|$ with probability at least $1-\frac{1}{n^{10}}$.
\item[(ii)] If $|E_M^a|\leq \frac{\vareps}{100}|T|$, $|S|\leq \frac{60 \log n}{\vareps^2} $ with probability at least $1-\frac{1}{n^{10}}$.
\end{itemize} 
\end{lem}
\begin{proof}
We use the similar argument as that of in Section~\ref{sec:e-known-vadeg} to show $\widehat{m}$ is an $(1\pm \vareps)$-approximation of $|E_M|$.

Here, $\mu=\E[|S|]=\frac{3000\log n}{\vareps^3 T}|E_m^a|$. We prove (i) and (ii) separately.
\begin{description}
\item[(i)] As $|E_M^a|\geq \frac{\vareps}{100}T$, $\E[|S|] \geq \frac{30 \log n}{\vareps^2 }$. Applying Lemma~\ref{lem:cher_bound1} (i) and (ii), 
\begin{eqnarray*}
\pr\left(\size{|S|-\E[|S|]}\geq \vareps \E[|S|]\right) &\leq& 2\exp{\left(-\frac{\vareps^2 \E[|S|]}{3}\right)}\\
 \pr\left(\size{\frac{\vareps^3T}{3000\log n}|S|-\size{E_M^a}}\geq \vareps |E_M^a|\right) &\leq& \frac{1}{n^{10}}.
\end{eqnarray*}
Observe thet we are done with the claim.
\item[(ii)] As $|E_M^a|\leq \frac{\vareps}{100}T$, $\E[|S|] \leq \frac{30 \log n}{\vareps^2 }$. Applying Lemma~\ref{lem:cher_bound1} (iii), by taking $t=\frac{30 \log n}{\vareps^2 }$ and $\delta=1$, we have 
\begin{eqnarray*}
\pr\left(\size{|S| \geq (1+\delta)t}\right) &\leq& \exp{\left(-\frac{\delta^2 t}{3}\right)} \\
\pr\left(|S| \geq \frac{60 \log n}{\vareps^2}\right) &\leq& \frac{1}{n^{10}}.
\end{eqnarray*}
Observe that, we are done with the claim.
\end{description}

\end{proof}

\section{Proof of Lemma~\ref{lem:cher-random-algo}}\label{sec:deviation}
\begin{lem}[Restatement of Lemma~\ref{lem:cher-random-algo}]
\begin{description}
\item[]
\item[(i)] For each $j \in [t]$ with $\size{B_j} \geq \frac{\sqrt{\vareps T}}{10t}$, $\pr \left( \size{{\size{B_j \cap R}}-
\frac{\size{R}\size{B_j}}{n}} \geq \frac{\vareps}{10}\frac{\size{R}\size{B_j}}{n}  \right) \leq \frac{1}{n^{10}}$.
\item[(ii)] For each $j \in [t]$ with $\size{B_j} < \frac{\sqrt{\vareps T}}{10t}$, $\mathbb{P} \left( {{\size{B_j \cap R}}} \geq \frac{\size{R}}{n}\frac{\sqrt{\vareps T}}{8t}  \right) \leq \frac{1}{n^{10}}$.
\item[(iii)] For each vertex $v_i$ with $d_M(v_i) \geq \frac{\sqrt{\vareps T}}{10t}$, $\mathbb{P} \left( \size{\kappa_{v_i} - \frac{\size{R}d_M(v_i)}{n}} \geq \frac{\vareps}{10}\frac{\size{R}d_M(v_i)}{n}\right) \leq \frac{1}{n^{10}}$.
\item[(iv)] For each vertex $v_i$ with $d_M(v_i) < \frac{\sqrt{\vareps T}}{10t}$, $\mathbb{P} \left( {\kappa_{v_i}} \geq \frac{ \size{R}}{n}\frac{\sqrt{\vareps T}}{8t}\right) \leq \frac{1}{n^{10}}$.
\end{description}
\end{lem}

\begin{proof} 
Let us take $N=n, r= \size{R} = \Gamma=\widetilde{\Theta}\left(\frac{n}{\sqrt{T}}\right)$ , $I=\{v_1,\ldots,v_n\}$ in Lemma~\ref{lem:without-replace}.
\begin{description}
\item[(i)] Setting $J=B_j$ and $\delta=\frac{\vareps}{10}$ in Lemma~\ref{lem:without-replace} (i) and (ii), we have 
$$ \pr \left( \size{{\size{B_j \cap R}}-
\frac{\size{R}\size{B_j}}{n}} \geq \frac{\vareps}{10}\frac{\size{R}\size{B_j}}{n}  \right) \leq 2\exp{\left(-\frac{(\vareps /10)^2\size{B_j}{\Gamma}}{3n}\right)}\leq \frac{1}{n^{10}}.$$
The last inequality holds as $\size{B_j} \geq \frac{\sqrt{\vareps T}}{10t}$, $t=\lceil \log _{1+\frac{\vareps}{10}}n\rceil=\Theta\left(\frac{\log n}{\vareps}\right)$ and $\Gamma=\widetilde{\Theta}\left(\frac{n}{\sqrt{T}}\right)$.
\item[(ii)] Set $J=B_j$, $k=\frac{\sqrt{\vareps T}}{10t}$, $\delta=\frac{1}{4}$ in Lemma~\ref{lem:without-replace} (iii). $\mbox{As}~\size{B_j} \leq\frac{\sqrt{\vareps T}}{10t},~\size{J} \leq k$. Hence, 
$$\mathbb{P} \left( {{\size{B_j \cap R}}} \geq \frac{\size{R}}{n}\frac{\sqrt{\vareps T}}{8t}  \right)\leq \exp{\left(-\frac{(1/4)^2 ({\sqrt{\vareps T}}/{10t})\Gamma}{3n}\right)}\leq \frac{1}{n^{10}}.$$
\item[(iii)] Setting $J$ as the set of monochromatic neighbors of $v_i$ in $R$ and $\delta=\frac{\vareps}{10}$ in Lemma~\ref{lem:without-replace} (i) and (ii), we get
$$\mathbb{P} \left( \size{\kappa_{v_i} - \frac{\size{R}d_M(v_i)}{n}} \geq \frac{\vareps}{10}\frac{\size{R}d_M(v_i)}{n}\right) \leq \exp{\left(-\frac{(\vareps /10)^2|J|\Gamma}{3n}\right)} \leq \frac{1}{n^{10}}.$$
The last inequality holds as $\size{J}=d_M(v_i) \geq \frac{\sqrt{\vareps T}}{10t}$, $t=\lceil \log _{1+\frac{\vareps}{10}}n\rceil=\Theta\left(\frac{\log n}{\vareps}\right)$ and $\Gamma=\widetilde{\Theta}\left(\frac{n}{\sqrt{T}}\right)$.
\item[(iv)] Set $J$ as the set of monochromatic neighbors of $v_i$ in $R$, $k=\frac{\sqrt{\vareps T}}{10t}$, $\delta=\frac{1}{4}$ in Lemma~\ref{lem:without-replace} (iii). Note that $\size{J}=d_M(v_i)\leq \frac{\sqrt{\vareps T}}{10t}=k $. Hence,
$$\mathbb{P} \left( {\kappa_{v_i}} \geq \frac{ \size{R}}{n}\frac{\sqrt{\vareps T}}{8t}\right)\leq \exp{\left(-\frac{(1/4)^2 ({\sqrt{\vareps T}}/{10t})\Gamma}{3n}\right)} \leq \frac{1}{n^{10}}.$$

\end{description}

\end{proof}





\section{Communication Complexity}\label{sec:comm}

Communication Complexity~\cite{DBLP:books/daglib/0011756} deals with finding the minimum amount of bits that is needed to communicate in order to compute a function when the input to the function is distributed among multiple parties. For the purpose of our work, we are concerned with two player games with one-way communication protocol. The players are traditionally called Alice and Bob. Both of them have a $n$-bit input string and are unaware of each other's input. The goal is to minimize the bits Alice needs to communicate to Bob so that he can compute a function on both their inputs. No assumption is made on their computational powers and there is no restriction on the amount of time needed for computing the function. Randomized one way communication complexity of a function, is defined as the number of bits sent by Alice, in the worst case, by the best randomized protocol to compute $fS$.

\subsection{INDEX problem in the communication complexity model}

Usually, the space lower bound results in the streaming model of computation are proved by a reduction from a problem in communication complexity. We establish our lower bounds by considering a reduction from the INDEX problem in the one-way communication protocol for two players to the specific problem in graphs in the \va model. The INDEX problem is defined as follows:
There are two parties, Alice and Bob. Alice has a $N$-bit input string $X \in \{0,1\}^{N}$ and Bob has an integer $j \in [N]$. Both are unaware of each other's input, but have an access to a public randomness, and the goal of Bob is to compute $X_{j}$, the $j^{th}$ bit of $X$, by receiving a single message from Alice. 
\begin{lem}~\cite{DBLP:books/daglib/0011756}\label{index}
The randomized one-way communication complexity of INDEX is $\Omega(N)$
\end{lem}

\section{Lower bounds for \conflictest} \label{sec:est-lowerbound}
 
We show a tight lower bound of $\Omega\left(\min\{\size{V},\frac{|V|^2}{T}\}\right)$ for the \conflictest problem in the vertex arrival model in Section~\ref{sec:est-vaaolbsc}. For the \conflictest problem in the vertex arrival with degree oracle model, we show a tight lower bound of $\Omega \left( \min \{ \size{V}, \frac{ \size{E} }{T} \} \right)$ in Section~\ref{sec:est-vadglbsc}. These bounds are proved using reductions from the \emph{INDEX} problem, discussed in Lemma~\ref{index} in Appendix~\ref{sec:comm}, in the one-way communication complexity model to the \conflictest problem in graphs (in the vertex arrival streaming models).

\remove{In Section~\ref{sec:lower-rand}, we show that a lower bound of $\Omega\left(\frac{n}{T^3}\right)$ for \conflictest in 
\varandlong via a reduction from the following variation of {\sc Multiparty Set Disjointness}. We discuss it in the respective section Section~\ref{sec:lower-rand} in details.}

\subsection{Lower bound for \conflictest in \va model} \label{sec:est-vaaolbsc}

\begin{theo}\label{theo:est-vaaolbth} Let $n,m,T  \in \N$ be such that $1 \leq T \leq {n \choose 2}$ and {}{$m\geq T$}. Any one pass streaming algorithm; that takes the {}{vertices and edges} of a graph $G(V,E)$ (with $\size{V}=\Theta(n)$ and $\size{E}=\Theta(m)$) and coloring function $f:V \rightarrow [C]$ on the vertices, in \va model; and determines whether the {number} of monochromatic edges in $G$ is {}{$0$} or {}{$T$} with probability $2/3$; requires 
$\Omega\left( \min \{n,\frac{n^2}{T}\}\right)$ bits of space.
 \end{theo}

\begin{proof} We show that the lower bound is $\Omega(n)$ when $T \leq \frac{n}{2}$ and $\Omega\left(\frac{n^2}{T}\right)$ when $T>\frac{n}{2}$, separately, to get the stated lower bound. We give a reduction from the INDEX problem to the \conflictest problem in graphs with $\Theta(n)$ vertices, $\Theta(m)$ edges and having at least $T$ conflicting edges, in the vertex arrival model. We show our reduction when {}{$m=T$}, but we can modify it for any {}{$m\geq T$}.

 The reduction works as follows. For $T \leq \frac{n}{2}$, Alice has an $N$-bit input string {$X \in \{0,1\}^{N}$}. For each input bit $X_i$, Alice creates a vertex {}{$p_i$}. If $X_i$ equals $1$, then vertex {}{$p_i$} is colored with color $C_1$, else it is colored with color $C_0$. After processing all bits of her input, Alice sends the current memory state to Bob. 
 Let $j\in [N]$ be the input of Bob. Bob constructs a gadget {}{$Q$} which is an independent set of $\left(n-N\right)$ vertices and colors all the vertices in the gadget with color $C_{1}$. He adds all the edges from the vertex {}{$p_{j}$} to the gadget {}{$Q$}. The number of vertices in the graph is {}{$(N) + (n-N) = n$ and the number of edges in the graph is $m = T$}. We set $N=n-T$. If $X_{j} = 0$, then the color of {}{$p_{j}$} is $C_{0}$ and there are {}{$0$} conflicting edges, where as if $X_{j} = 1$, then the color of {}{$p_{j}$} is $C_{1}$ and there will be {}{$T$} conflicting edges. Therefore, for $N = n-T$, deciding whether the number of monochromatic edges in the graph is $0$ or $T$, requires $\Omega(n-T)$ or $\Omega(n)$ space.

For $T\geq  \frac{n}{2}$, Alice has an $N$-bit input string $X \in \{0,1\}^{N}$. For each input bit $X_i$, Alice constructs an independent set {}{$P_{i}$} of size $\frac{2T}{n}$. If $X_i$ equals $1$, then the vertices of {}{$P_{i}$} are colored with color $C_1$, else the vertices are colored with color $C_0$. After processing all bits of her input, Alice sends the current memory state to Bob. Let $j\in [N]$ be the input of Bob. Let $j\in [N]$ be the input of Bob. Bob constructs a gadget {}{$Q$} which is an independent set of $\frac{n}{2}$ vertices and colors all the vertices with the color $C_{1}$. He adds all the edges from the gadget {}{$P_{j}$} to the gadget {}{$Q$}. We set $N = \frac{n^{2}}{T}$. The number of vertices in the graph is {}{$\frac{2T}{n} \cdot( N) + \frac{n}{2} = \Theta(n)$ and the number of edges in the graph is $m = T$}.
If $X_{j} = 0$, then the color of vertices in {}{$P_{j}$} is $C_{0}$ and there are {}{$0$} conflicting edges, where as if $X_{j} = 1$, then the color of vertices in {}{$P_{j}$} is $C_{1}$ and there will be {}{$T$} conflicting edges. Therefore, for $N = \frac{n^{2}}{T}$, deciding whether the number of monochromatic edges in the graph is $0$ or $T$, requires $\Omega \left(\frac{n^2}{T} \right)$ space.

Recall that we are doing our reductions for {}{$m=T$}. We make the above constructions work for any {}{$m\geq T$} by adding a complete subgraph on {}{$\sqrt{m-T}$} vertices such that none of the edges of the complete subgraph are conflicting.~\footnote{Note that $\sqrt{m-T}=O(n)$.}
\end{proof} 
 
\subsection{Lower bound for \conflictest in \vadeg model} \label{sec:est-vadglbsc}

\begin{theo}\label{theo:est-vadglbth} Let $n,T  \in \N$ be such that $1 \leq T \leq {n \choose 2}$. Then there exists an $m$ with $T\leq m \leq {n \choose 2}$ such that the following happens. Any one pass streaming algorithm; that takes the {}{vertices and} edges of a graph $G(V,E)$ (with $\size{V}=\Theta(n)$ and $\size{E}=\Theta(m)$) and a coloring function $f:V \rightarrow [C]$ on the vertices, in \vadeg model; and determines whether the number of monochromatic edges in $G$ is {}{$0$} or {}{$T$} with probability $2/3$; requires 
$\Omega\left( \min \{n,\frac{m}{T}\}\right)$ bits of space.
 \end{theo}

\begin{proof} We show that the lower bound is $\Omega(n)$ when $m > nT$ and $\Omega\left(\frac{m}{T}\right)$ when $m \leq nT$, separately, to get the stated lower bound. We give a reduction from the INDEX problem to the \conflictest problem in graphs with $\Theta(n)$ vertices, $\Theta(m)$ edges and having atleast $T$ conflicting edges, in the vertex arrival model. The existence of $m$ will be evident from the construction.

 The reduction works as follows. For $m > nT$, Alice has an $N$-bit input string $X \in \{0,1\}^{N}$. For each input bit $X_i$, Alice creates a vertex {}{$p_i$}. If $X_i$ equals $1$, then vertex {}{$p_i$} is colored with color $C_{i1}$, else it is colored with color $C_{i0}$.  After processing all bits of her input, Alice sends the current memory state to Bob. 
 Let $j \in [N]$ be the input of Bob. Bob constructs a gadget {}{$Q$} of {}{$n-N$} vertices such that {}{$Q$} is an independent set of $\left(n-N\right)$ vertices. Bob colors all the vertices in the gadget {}{$Q$} with color $C_{j1}$. He adds all the edges from {}{$Q$} to all the vertices in {}{$\{l_i:i \in [N]\}$}. The number of vertices in the graph is {}{$N+ (n-N) = \Theta(n)$ and the number of edges in the graph is $m = NT$}. We set $N=n-T$. If $X_{j} = 0$, then the color of {}{$p_{j}$} is $C_{j0}$ and there are {}{$0$} conflicting edges, where as if $X_{j} = 1$, then the color of {}{$p_{j}$} is $C_{j1}$ and there will be {}{$T$} many conflicting edges. Therefore, for $N = n-T$, deciding whether the number of monochromatic edges in the graph is $0$ or $T$, requires $\Omega(n-T)$ or $\Omega(n)$ space. Observe that, the degree of the vertices are independent of the inputs of Alice and Bob. In particular, the degree of every vertex in {}{$\{p_i:i \in [N]\}$} is {}{$\size{Q}=n-N$} and the degree of every vertex in {}{$Q$} is {}{$N$}. So, the availability of degree oracle will not help in the above construction.

For $m \leq nT$, Alice has an $N$-bit input string $X \in \{0,1\}^{N}$. For each input bit $X_i$, Alice constructs an independent set {}{$P_{i}$} of size $\frac{2nT}{m}$. If $X_i$ equals $1$, then the vertices of {}{$P_{i}$} are colored with color $C_{i1}$, else the vertices are colored with color $C_{i0}$. After processing all bits of her input, Alice sends the current memory state to Bob. 
Let $j \in [N]$ be the input of Bob. Let $j\in [N]$be the input of Bob. Bob constructs a gadget {}{$Q$} where {}{$Q$} is an independent set of $\frac{m}{2n}$ vertices. Bob colors the vertices in {}{$Q$} with $C_{j1}$ and he adds all the edges from {}{$Q$} to 
{}{$(P_1\cup \cdots \cup P_N)$}. We set $N = \frac{m}{T}$. The number of vertices in the graph is {}{$\frac{2nT}{m} \cdot N + \frac{m}{2n} =\Theta(n)$} as $T \geq \frac{n}{2}$ and $m \geq T$ {}{and the number of edges in the graph is $m = NT$}.
If $X_{j} = 0$, then the color of vertices in {}{$P_{j}$} is $C_{j0}$ and there are {}{$0$} conflicting edges, where as if $X_{j} = 1$, then the color of vertices in {}{$P_{j}$} is $C_{j1}$ and there will be {}{$T$} conflicting edges. Therefore, for $N = \frac{m}{T}$, deciding whether the number of monochromatic edges in the graph is $0$ or $T$,  requires $\Omega \left(\frac{m}{T} \right)$ space.
Observe that, the degree of the vertices are independent of the inputs of Alice and Bob. In particular,  the degree of every vertex in 
{}{$P_1\cup \cdots \cup P_N$} is {}{$\size{Q}=\frac{m}{2n}$}, and the degree of every vertex in {}{$Q$} is {}{$N\cdot \frac{2Tn}{m}$}. So, the availability of degree oracle will not help in the above construction. 
\end{proof}

\end{document}